\documentclass[11pt]{iopart}
\usepackage{amssymb,iopams,bm,verbatim,enumerate,perpage}
\usepackage[T1]{fontenc}
\usepackage{amsthm}
\usepackage{color}
\makeatletter
\gdef\@ctrerr{%
  \@latex@warning{Counter too large}}
\makeatother
\MakePerPage[2]{footnote}
\def\Journal#1#2#3#4{{#4} {\it #1} {\bf #2}, #3 }
\def\ud{\textrm{d}}

\newcommand{\smfrac}[2]{{\textstyle{#1\over#2}}}
\def\half{\tfrac{1}{2}}
\newcommand{\w}[1]{\bm{#1}} 
\def\h{h}
\def\u{\dot{u}}
\def\uc{\overline{\dot{u}}}
\def\zc{\overline{z}}
\def\U{\dot{U}}
\def\modU{\mathcal{U}}
\def\B{\mathfrak{B}} 
\def\E{\mathfrak{E}} 
\def\Q{\mathfrak{Q}} 
\def\R{\mathfrak{R}} 
\def\b{\mathfrak{b}} 
\def\O{\mathcal{O}} 
\def\N{\mathfrak{N}} 
\def\J{\mathcal{J}} 
\def\X{X} 
\def\Y{Y} 
\def\Z{Z} 
\def\kappac{\overline{\kappa}}
\def\tauc{\overline{\tau}}
\def\pic{\overline{\pi}}
\def\rhoc{\overline{\rho}}
\def\betac{\overline{\beta}}
\def\ez{\epsilon_0}
\def\d{\delta}
\def\dc{\overline{\delta}}
\def\Eone{\mathcal{E}_1}
\def\Eonec{\overline{\mathcal{E}_1}}
\def\Etwo{\mathcal{E}_2}
\def\Etwoc{\overline{\mathcal{E}_2}}
\def\Ethree{\mathcal{E}}
\def\F{\mathcal{F}}
\def\BB{\mathfrak{B}}
\newcommand{\et}{\eth}
\newcommand{\etd}{\eth '}
\def\ZZ{\mathcal{Z}}
\def\tfrac{\smfrac}

\newtheorem*{de}{Definition}
\newtheorem{re}{Remark}
\newtheorem{lm}{Lemma}
\newtheorem{te}{Theorem}

\newtheorem{co}{Corollary}
\begin{document}

\title[Shear-free perfect fluids]{Shear-free perfect fluids with a barotropic equation of state in general relativity: an update and a new formalism.}
\author{Norbert Van den Bergh}
\address{Belgisch Onderzoekscentrum voor Relativiteitstheorie en Kosmologie/ \\ Institut Belge de Recherche en Relativit\'{e} et Cosmologie,\\
 Holleweg, Molenstede, Belgium}
\eads{\mailto{norbert.vandenbergh@gmail.com}}

\begin{abstract}
The present status of the shear-free perfect fluid conjecture in general relativity is discussed: I give a review of recent work and present a GHP-like formalism which may provide a better understanding of the problems encountered. The use of the formalism is illustrated by means of several examples, among which there are some new results, as well as two short alternative proofs for the classical theorems, where the vorticity and acceleration are parallel or where the expansion and matter density are functionally related.
\end{abstract}


\section{Introduction}
I consider perfect fluid solutions of the Einstein field equations\footnote{a large part of this introduction has been taken over from \cite{VdB-Radu2016}},
\begin{equation}
R_{ab} - \tfrac{1}{2} R \, g_{ab} = T_{ab}
, \label{EFE}
\end{equation}
on a 4-dimensional spacetime $(M, g)$, with energy-momentum tensor given by
\begin{equation}\label{intro_Tab}
T_{ab}= (\mu+p) u_{a}u_{b}+p g_{ab},
\end{equation}
$\mu$ and $p$ being respectively the energy density and pressure of the fluid and the unit time-like vector field $u_a$ being the fluid's (covariant) 4-velocity. As is well known,
the covariant derivative of $u_a$ can be decomposed as
\begin{equation}\label{uadecomp}
u_{a;b}=\tfrac{1}{3}\theta (g_{ab}+u_a u_b)+\sigma_{ab}+\omega_{ab}-\dot{u}_{a}u_{b},
\label{ucoveq2}
\end{equation}
where $\theta$ is the fluid's (rate of volume) expansion, $\dot{u}_{a}$ is the acceleration and $\sigma_{ab}$, $\omega_{ab}$ are respectively the
shear and vorticity tensors, which are uniquely defined by (\ref{ucoveq2}) and the properties
\begin{equation}
 u^a \dot{u}_{a} = 0,\ u^a\omega_{ab} = \ u^a\sigma_{ab} =0, \ \sigma_{[ab]}=\omega_{(ab)}=0, \ {\sigma^a}_a=0.
\end{equation}
The physical significance of these so called \emph{kinematical quantities} has been discussed by many authors, see for example \cite{Ellis1971}. Among
well known explicit solutions of the Einstein field equations \cite{Kramer} one notes the following \textit{shear-free} ($\sigma_{ab}=0$) ones: the Einstein static
universe ($\theta=\dot u_a=\sigma_{ab}=\omega_{ab}=0$), the FLRW universes ($\theta \neq 0$, $\dot u_a=\sigma_{ab}=\omega_{ab}=0$) and the G\"odel universe ($\theta = \dot u_a = \sigma_{ab}=0$, $\omega_{ab} \neq 0$).
Imposing a barotropic equation of state (EOS) $p = p(\mu)$ leads to extra restrictions on the solution space: for example all barotropic and shear-free ($\sigma_{ab}=0$) perfect fluids with non-vanishing expansion and vanishing vorticity are known explicitly \cite{CollinsKV}. This is not the case for barotropic and shear-free perfect fluids with vanishing expansion and non-vanishing vorticity \cite{Karimian2012}, although also here large classes of solutions
exist (for example all rigidly rotating axisymmetric and stationary perfect fluids belong to this family). Remarkably, barotropic and shear-free perfect fluids in which both expansion and vorticity are non-zero seem to be confined to the limiting situation of `$\Lambda$-type' models (meaning $p=-\mu=constant$), the only example known so far being a Bianchi IX model found by Obukhov et al.~\cite{Obukhov}. This  brings us to the subject of the present paper, the so-called \emph{shear-free fluid conjecture} which claims that
\begin{quote}
general relativistic, shear-free perfect fluids obeying a barotropic EOS such that $p + \mu \neq  0$, are either non-expanding
or non-rotating.
\end{quote}
During the last few years
there has been a renewed interest in this conjecture, which, if true, would be a remarkable consequence of the full Einstein
field equations: on the one hand Newtonian perfect fluids with a barotropic EOS, which are rotating and expanding but non-shearing, are known
to exist \cite{Ellis2011,HeckmannSchucking, Narlikar,SenSopSze}, while on the other hand
in, for example, $f(R)$ gravity there is no counterpart of the conjecture either \cite{Sofuoglu}.
The conjecture is also specific to the Lorentzian signature, since in the Riemannian case the example of the Eguchi-Hanson (Ricci-flat) metric provides us with
a shear-free, expanding and rotating congruence \cite{Pantilie2002}.
\\
The first suggestion that the vanishing of shear
could play a decisively restricting role in the construction of expanding and rotating perfect fluids appeared in 1950, without proof, in a somewhat obscure contribution by
G\"odel \cite{Godel2} on homogeneous rotating cosmological models. A precise formulation of G\"odel's claim was given in 1957 by Sch\"ucking \cite{Schucking}, who gave a short coordinate based
proof that spatially homogeneous dust models ($p=0$) could be either rotating or expanding, but not both.  The condition of vanishing pressure was dropped by
Banerji \cite{Banerji}, who gave a similar coordinate based proof for (tilted) spatially homogeneous perfect fluids obeying a $\gamma$-\textit{law EOS},
$p=(\gamma-1)\mu$, with $\gamma - 1 \neq \frac{1}{9}$ \footnote{the reason why $\gamma - 1 = \frac{1}{9}$ is special was clarified in \cite{NorbertCQG1999}, where also a proof was given for non-spatially homogeneous spacetimes}. Sch\"ucking's result was generalized in 1967 by Ellis \cite{Ellis1}, who used
the orthonormal tetrad formalism to show that the restriction of spatial homogeneity was redundant for dust spacetimes, while in \cite{WhiteCollins} it
was observed that Ellis'
result remained valid in the presence of a cosmological constant. A covariant proof of both results was presented more recently in \cite{Sikhonde_Dunsby}.
In \cite{TreciokasEllis} Treciokas and Ellis proved, again using a combination of an orthonormal tetrad formalism and an adapted choice of coordinates,
that the conjecture held true also for the EOS $p=\frac{1}{3} \mu$, a result which was generalised by Coley \cite{Coley} to allow for a possible non-zero
cosmological constant. In \cite{TreciokasEllis} an outline of an argument was presented, indicating the validity of the conjecture for perfect fluids
in which the \textit{acceleration potential} $r = \exp \int_{p_0}^p \frac{1}{\mu+p} \ud p$ satisfies an equation of the form $\dot{r}= \beta(r)$, where the \textit{dot-operator} is the derivative along the fluid 4-velocity. This result
(which implies the validity of the
conjecture for a general EOS, once one additionally assumes spatial homogeneity, as was the case in \cite{Banerji, KingEllis, Whiteth}) will play a key role in the sequel.
However the details of the underlying proof remained veiled until 1988, when Lang and Collins \cite{Lang,Langth} explicitly showed that $\omega \theta=0$ indeed follows,
provided there exists a functional relation of the form $\theta=\theta(\mu)$ (which, by the conservation law $\dot{\mu} + (\mu+p) \theta=0$, is equivalent with $\dot{r}=\beta (r)$).
A `covariant' proof
of this same result was given by Sopuerta in \cite{Sopuerta1998}.
While Treciokas and Ellis already questioned the possible existence of rotating and expanding perfect fluids with
$p=p(\mu)$, their non-existence was explicitly conjectured by Collins \cite{CollinsKV}, following a series of papers in which the conjecture was proved successively for the cases where
the vorticity vector is parallel to the acceleration (see \cite{WhiteCollins}, or \cite{SenSopSze,Sikhonde_Dunsby,Sikhonde_thesis} for a fully covariant proof), or in which the Weyl tensor is purely magnetic \cite{Collins1984} or purely electric \cite{CyganowskiCarminati, Langth}.

Since then this so-called '\textit{shear-free perfect fluid conjecture}' has been proved also in a large number of special cases, such as
$\ud p/\ud \mu = -\frac{1}{3}$ \cite{CyganowskiCarminati, Langth, slob, NorbertCQG1999}; $\theta=\theta(\omega)$ \cite{Sopuerta1998};
Petrov types N \cite{Carminati1990} and III \cite{CarminatiCyganowski1996,CarminatiCyganowski1997}; the existence of a conformal Killing vector parallel to the fluid flow \cite{Coley}; the Weyl tensor having either a divergence-free electric part \cite{NorbertKarimianCarminatiHuf2012}, or a divergence-free magnetic
part, in combination with an EOS which is of the $\gamma$-law type \cite{NorbertCarminatiKarimian2007} or which is sufficiently
generic \cite{CarminatiKarimianNorbertVu2009}, and in the case where the Einstein field equations are linearised
about a FLRW background \cite{Nzioki} . 
A major step was achieved by Slobodeanu \cite{Slobodeanu2014} proving the conjecture for an arbitrary $\gamma$-law EOS
(except for the cases $\gamma-1 = -\tfrac{1}{5}, -\tfrac{1}{6},-\tfrac{1}{11},-\tfrac{1}{21},\tfrac{1}{15}, \tfrac{1}{4}$) and a vanishing cosmological constant. In this approach the Einstein field equations were seen as a second order differential system in the length scale function,
with the integrability conditions for the system allowing one
to prove the conjecture \emph{via} some sufficient conditions in terms of \emph{basic functions}, i.e.~functions
that are constant along the fluid flow. The
results of \cite{Slobodeanu2014} were generalzed in \cite{VdB-Radu2016}, covering the exceptional values of $\gamma$, while allowing also for a non-zero cosmological
constant, or, equivalently, generalising the EOS to the form $ p=(\gamma -1) \mu + p_0$, a so called \emph{linear EOS}.
\footnote{Inclusion of the cosmological constant is non-trivial only in the case of a $\gamma$-law EOS, while for a general EOS it can always be absorbed in the definitions of pressure and matter density.}
In \cite{VdB-Radu2016} the formalism of \cite{Slobodeanu2014} was adapted to a general EOS and some theorems were presented, which not only played a key role in the proof for a linear EOS, but which may also be useful when tackling
the conjecture in its full generality, when $p=p(\mu)\neq-\mu$ is an arbitrary function of the matter density. These
theorems tell us that,
provided certain algebraic restrictions are obeyed by the kinematical quantities, the conjecture is valid or that, if it does \emph{not} hold (and under different conditions), there exists a Killing vector along the vorticity. In the latter
case the equations describing the problem simplify dramatically, but the accompanying loss of information turns this sub-case, as remarked already by
Collins \cite{CollinsKV}, into an exceptionally elusive one. The simplest of these
criteria says that, for an expanding and rotating shear-free perfect fluid obeying a barotropic equation of state, the existence of a Killing vector along the vorticity is equivalent with the acceleration being orthogonal with the vorticity. A covariant proof of this result was presented recently in \cite{Sikhonde_thesis}, generalizing it furthermore to the property that an expanding and rotating barotropic shear-free perfect fluid admits a Killing vector along the vorticity, if the projection of the acceleration in the plane orthogonal to the vorticity has components $\u_1=\mathcal{F} \mathfrak{B}_1$ and $\u_2=\mathcal{F} \mathfrak{B}_2$, with $\mathcal{F}$ a function of the matter density $\mu$ and $\mathfrak{B}_1,\mathfrak{B}_2$ basic functions.\\

In the above proofs one generally starts with the assumption that an expanding and rotating barotropic shear-free perfect fluid exists and one tries to arrive at a contradiction leading to `$\omega \theta = 0$`. Needless to say, any calculational or conceptual error in this procedure is prone to lead to a false proof: examples are \cite{Carminati2015} (see \cite{VdB-Radu2016}) and  \cite{Goswami-Ellis} (see \cite{VdBCarminati2021}). The fact that most recent proofs require a fair amount of symbolic algebra (using Maple or Mathematica, together with some -usually not readily available- special purpose software), makes the spotting of errors quite difficult.

In order to arrive at the above contradiction two main approaches have been used (besides the early coordinate based methods), (a) the tetrad approach and (b) the covariant one. Tetrad approaches (using either an orthonormal or null-tetrad) have been intensively used, starting with Ellis's work on dust space-times \cite{Ellis1} and culminating in the introduction and efficient use of \emph{basic} quantities in \cite{Slobodeanu2014, VdB-Radu2016}\footnote{recent work by Carminati, using the Maple package STeM, is proceeding along these same lines and looks very promising (private communication)}.
As the vorticity vector $\w{\omega}$ and the velocity vector $\w{u}$ are orthogonal, it is natural to choose an orthonomal tetrad in which $\w{u}$ and $\w{\omega}$ fix the time-like basis vector $\w{e}_0$ together with one of the space-like basis vectors, say $\w{e}_3$ (in the case of a null-tetrad this would then fix the $(\w{k},\w{\ell})$ plane). Fixing the remaining basis vectors $\w{e}_1,\w{e}_2$ is a matter of taste and is often not necessary (or not even possible, if a partial isotropy would be present). This situation is parallelled in the covariant approaches, where also in the so-called `fully covariant' ones \cite{SenSopSze, Sopuerta1998} all essential information appears to have been obtained by projecting the relevant covariant equations on $\w{u}$, $\w{\omega}$ and (in \cite{Sopuerta1998}) on the spatial vector $\w{\omega} \times (\w{\u} \times \w{\omega})$.

A fully covariant approach appearing to be somewhat of an overhead, it looks more natural to use a 1+1+2 covariant formalism \cite{Clarkson2007} in which the 3-space orthogonal to $\w{u}$ is further decomposed w.r.t.~the vorticity vector. A further argument (which also applies to the 1+1+2 approach) for stepping away from full covariance is that it appears awkward then to split the equations in their \emph{basic} and \emph{non-basic} parts\footnote{for attempts in this direction, see \cite{Sikhonde_thesis}}. Therefore there seems to be  some advantage in introducing a `\emph{basic SO2 formalism}', in which all equations are invariant under \emph{basic rotations} (being rotations for which the rotation angle $\alpha$ satisfies $\w{e}_0(\alpha)$=0). Absorbing the (under this sub-group) `badly transforming' rotation coefficients in new (tetrad-)derivative operators, it follows that all variables, together with their derivatives, become explicitly covariant under basic rotations. A further advantage of such a formalism is that, allowing only basic rotations, the tetrad still can be chosen to be co-rotating with the fluid, a choice which has proven its usefulness since the early tetrad approaches. The whole procedure is very reminiscent of the way in which, starting from the Newman-Penrose formalism, the Geroch-Held-Penrose formalism \cite{GHP} was constructed, or the ortho-complex-null (OCN) formalism \cite{Wylleman_thesis} (see also \cite{Wylleman_NV_2007,Wylleman_ConfProc2009,Wylleman_MG2012}) from the 1+1+2 (tetrad) formalism, but now enabling one to deal in an efficient way with the concept of \emph{basic} quantities. It remains to be seen whether the formalism may be adapted to be useful also for the study of more general (shearing) perfect fluids.

I begin with introducing in section \ref{sectionconventions} the necessary notations and conventions, while in section \ref{section_recaptheorems} I give an overview of old and recent results in terms of the new formalism.
Some new results are discussed in sections \ref{section_Muzitheorem1} and \ref{section_app2}.

\section{Notations and fundamental equations}\label{sectionconventions}

I first recapitulate the main results of \cite{VdB-Radu2016} by introducing at each point of spacetime an orthonormal tetrad $(\w{e}_a)=(\w{e}_0, \w{e}_\alpha)$ with
the time-like unit vector $\w{e}_0$
coinciding with the fluid 4-velocity $\w{u}$ (henceforth
Latin indices are tetrad indices taking the values 0,1,2,3, while Greek indices are spatial triad indices taking the values 1, 2, 3).
Boldface symbols always refer to vector (tensor) fields, but for readability (and as is customary in the literature, see e.g.~\cite{EllisMaartensmacCallum2012}) we
will also write
$\w{e}_a=\partial_a$: for example $\w{u}=\partial_0$, $\dot{\w{u}}=\dot{u}^\alpha \partial_\alpha$, $\dot{\w{u}}^2 = \dot{u}_\alpha \dot{u}^\alpha$ etc. ... \\
The volume 4-form components will be denoted by $\eta_{abcd}$ with the convention $\eta_{0123}=-1$; its restriction to tangent hyperplanes orthogonal
to $\w{u}$ is $\varepsilon_{\alpha\beta\gamma}$.
To a space-like 2-form one associates a vector field by Hodge duality, e.g.~the vorticity vector
$\w{\omega}$ has components $\omega_\alpha = \tfrac{1}{2}\varepsilon_{\alpha\beta\gamma}\omega^{\beta\gamma}$.  The notation $\omega$ will stand for the norm of the vorticity vector.

The metric components are $(g_{ab})=\mathrm{diag}(-1,1,1,1)$ and the Riemann and Ricci curvature tensors respectively satisfy
\begin{equation}\label{ricciv}
{v^a}_{;d;c}-{v^a}_{;c;d} = {R^a}_{bcd}v^b \ , \quad  R_{ab}={R^m}_{amb},
\end{equation}
while the 'trace-free part' of the curvature, given by the Weyl tensor, is
\begin{equation}\label{Weyldef}
C_{abcd}=R_{abcd} - (g_{a[c}R_{d]b}+g_{b[d}R_{c]a})+\tfrac{1}{3}R \, g_{a[c}g_{d]b}.
\end{equation}

The starting point is an \emph{extended orthonormal tetrad formalism} \cite{Norbert2013}, see also \cite{EllisMaartensmacCallum2012, MacCallum1971}, in which the
\textit{main variables} are
\begin{itemize}
\item the tetrad basis vectors $\partial_a$,

\item the kinematical quantities $\dot{u}_\alpha$, $\omega_\alpha$, $\theta$, $\sigma_{\alpha \beta}$,

\item the local angular velocity $\Omega_\alpha$ of the triad $\partial_{\alpha}$ with respect to a set of Fermi-propagated axes and the Kundt-Sch\"ucking-Behr variables \cite{MacCallum1971} $a_{\alpha}$ and $n_{\alpha \beta}=n_{\beta \alpha}$ which parametrize the purely spatial commutation coefficients ${{\gamma}^\alpha}_{\beta\kappa}$. They are defined by the relations\footnote{Note that in the definitions below (following MacCallum \cite{MacCallum1971}) the sign of $\Omega_\alpha$ is opposite to that in \cite{EllisMaartensmacCallum2012} }
\begin{eqnarray}
{}[\partial_0,\partial_\alpha] &=& \dot{u}_\alpha \partial_0 - \left(\tfrac{1}{3} \theta\delta_\alpha^{\beta}+\sigma_\alpha^{\beta}
+{\varepsilon^\beta}_{\alpha\gamma}(\omega^\gamma+\Omega^\gamma)\right) \partial_\beta~, \label{comm_explicit} \\
{}[\partial_\alpha, \partial_\beta ] &=& {{\gamma}^c}_{\alpha\beta}\partial_c\equiv -2\varepsilon_{\alpha\beta\gamma}\omega^\gamma\partial_0 + \left(2a_{[\alpha}\delta^\gamma_{\beta]}+
\varepsilon_{\alpha\beta\delta} n^{\delta\gamma}\right)\partial_\gamma~.\label{comm_explicit_bis}
\end{eqnarray}
It is computationally advantageous however to replace $a_{\alpha}$ and $n_{\alpha\beta}$ $(\alpha\neq\beta)$ with new variables $q_{\alpha}$ and $r_{\alpha}$ defined by
\begin{equation*}
n_{\alpha-1 \,\alpha+1}=(r_{\alpha}+q_{\alpha})/2, \quad a_{\alpha}=(r_{\alpha
}-q_{\alpha})/2.
\end{equation*}

\item the energy density $\mu$ and pressure $p$,

\item the `electric' and `magnetic' parts $E_{\alpha \beta}$, $H_{\alpha \beta}$ of the Weyl tensor with respect to $\w{u}$:
\begin{equation}
 E_{ab}= C_{acbd}u^c u^d , \quad
 H_{ab}= \frac{1}{2} \eta_{amcd} {C^{cd}}_{bn}u^m u^n, \label{EHdef}
\end{equation}
which are symmetric trace-free tensors determining the Weyl curvature.
\end{itemize}
In addition, I shall use the following \textit{auxiliary variables}: the spatial gradient of the expansion scalar, $z_\alpha=\partial_\alpha
\theta$, and the (covariant) divergence of the acceleration, $ j\equiv {\dot{u}^a}_{;a} =
\partial_{\alpha}\dot{u}^{\alpha}+\dot{u}^{\alpha}\dot{u}_{\alpha}-2 \dot{u}^{\alpha} a_{\alpha}$.\\

Note that with this choice of variables, once one assumes the Einstein equations (\ref{EFE}) to be satisfied, the Riemann tensor is
actually \emph{defined} in terms of $\w{E},\w{H},p$ and $\mu$ via (\ref{Weyldef}, \ref{EHdef}) \footnote{for example $R_{1212}=\tfrac{1}{3}
\mu - E_{33}$}, with the symmetry
and trace-free properties of $\w{E}$ and $\w{H}$ guaranteeing the usual symmetry properties of a curvature tensor.
The usual defining formulae (obtained from the second Cartan structure equations or, equivalently, from (\ref{ricciv})),
\begin{equation}\label{cartan2bis}
{R^a}_{bcd}={\Gamma ^a}_{bd,c}-{\Gamma ^a}_{bc,d}+{\Gamma ^e}_{bd}
{\Gamma ^a}_{ec} - {\Gamma ^e}_{bc} {\Gamma ^a}_{ed} - {{\gamma}^e}_{cd} {\Gamma ^a}_{be}~,
\end{equation}
become then a set of first order partial differential \emph{equations} in the connection coefficients
${\Gamma ^c}_{ab}$, which are related to the main variables of the formalism through the commutation coefficients:
\begin{equation}\label{commdef2}
\Gamma_{\ ab}^c = \tfrac{1}{2}\left(\gamma_{\ ba}^c + \gamma_{\ cb}^a - \gamma_{\ ac}^b \right).
\end{equation}

This set of equations (\ref{cartan2bis}) is automatically satisfied \cite{Norbert2013} if one takes as governing equations
the following system:
\begin{enumerate}[i)]
\item the Einstein field equations (\ref{EFE}),
\item the Jacobi equations $\left[\partial_{[a},\left[\partial_b,\partial_{c]}\right]\right]=0$, or
\begin{equation}\label{Jacobibis}
\partial_{[a}{{\gamma}^d}_{bc]}-{{\gamma}^d}_{e[a} {{\gamma}^e}_{bc]} =0~,
\end{equation}
\item 18\footnote{three of which are identities under the Jacobi equations} Ricci equations
${u^a}_{;d;c}-{u^a}_{;c;d}= {R^a}_{0cd}$ and
\item 20 Bianchi equations $R^a{}_{b[cd;e]} = 0$,
\end{enumerate}
where the $R_{ab}$ components in $(i)$ are replaced, via (\ref{cartan2bis}), in terms of
commutation coefficients ${{\gamma}^a}_{bc}$ and their derivatives.

\paragraph{Tetrad conventions.} It has become customary \cite{WhiteCollins} to
align $\partial_{3}$ with $\w{\omega}$, such that $\w{\omega} =\omega\partial_{3}\neq 0$.
Applying the commutators $[\partial_3, \partial_\alpha]$ to $p$ and using the Euler and Jacobi equations, one can show that the spatial triad can be taken to be co-rotating: $\w{\Omega}+\w{\omega}=0$, with
the remaining tetrad freedom consisting of rotations in the $(1,2)$ plane,
\begin{equation}\label{basicrotations}
\partial_1+ i\,\partial_2 \to e^{i\alpha} (\partial_1+i\,\partial_2),\textrm{ with } \partial_0\alpha =0.
\end{equation}
In accordance with the definition of basic variables (see section 3), I will call such transformations \emph{basic rotations}. Notice that, under $\partial_1+i\,\partial_2 \to e^{i\alpha} (\partial_1+ i\,\partial_2)$,
\begin{equation*}
\tfrac{1}{2}(n_{11}-n_{22})+ i \, n_{12} \longrightarrow e^{2 i \alpha} \left(\tfrac{1}{2}(n_{11}-n_{22})+ i \, n_{12}\right),
 \end{equation*}
while under $\sigma_{ab}=0$ and $\w{\Omega}+\w{\omega}=0$ the evolution equations for $n_{11} -n_{22}$ and $n_{12}$ are identical: it therefore follows that it is \emph{possible} to specialize the tetrad by means of a basic rotation so as to
achieve $n_{11}=n_{22}\equiv n$ (which would fix the tetrad, unless
\begin{equation}\label{extra_rot}
n_{12}=n_{11}-n_{22}=0,
\end{equation}
in which case further basic rotations can be used to obtain extra simplifications). Although in \cite{VdB-Radu2016} the choice $n_{11}-n_{22}=0$ was made throughout, I will purposely refrain from this here, in order to obtain a rotation invariant formalism.

\paragraph{Conventions related to the EOS.} Throughout I assume $p=p(\mu)$ with $p' :=\ud p / \ud \mu \neq -1$, excluding thereby the $\Lambda$-type models. I also restrict the formalism to the situation where $p' \neq 0 \neq p' + \tfrac{1}{3}$, as denominators of this form will appear throughout in order to simplify the equations. For the exceptional cases $p'=0$ or $p' + \tfrac{1}{3}=0$ the non-existence of shearfree rotating and expanding perfect fluids has been demonstrated by several authors(see the introduction).
I also adopt the notations:
\begin{eqnarray}\label{EOSconventions}
G &\equiv&  \frac{p''}{p'}(p+\mu) -p'+\frac{1}{3}, \\
G' &=& \ud G / \ud \mu, \ G_p=G'/p',\\
\lambda &=& \exp \int \frac{\ud \mu}{3(p+\mu)}.
\end{eqnarray}

Imposing in the extended tetrad formalism the existence of a barotropic equation of state $p=p(\mu)$, as well as the vanishing of the shear, results in new chains of integrability conditions. The procedure of building up the sequence of integrability conditions has been
carried out in several papers and for details of their derivation I refer the reader for example to
\cite{NorbertKarimianCarminatiHuf2012} or \cite{VdB-Radu2016}; see also \cite{Norbert2013}, or \cite{MaartensBassett1998} for the compact `1+3 covariant form' of some of these equations.

\section{Formulation in terms of basic variables}
An important role in recent proofs has been played by so-called basic objects (cf.~\cite{Slobodeanu2014} and references therein).
Let $\mathcal{H}$ denote the space-like subspace of the tangent space, orthogonal to the velocity $\w{u}$. Indicating the component along $\mathcal{H}$ or the restriction to $\mathcal{H}$ by a superscript, a tensorial object $\varsigma$ in $(\otimes^r \mathcal{H}) \otimes (\otimes^s \mathcal{H}^*)$ is called \emph{basic} if $(\mathcal{L}_{\w{u}} \varsigma)^\mathcal{H} =0$, $\mathcal{L}$ denoting the Lie derivative. In particular,
\begin{de}
A function $f$ on $M$ is \emph{basic} if it is conserved along the flow, $\w{u}(f)=0$, and a vector field $\w{X}$ belonging to $\mathcal{H}$ is \emph{basic} if $[\w{u}, \w{X}]^\mathcal{H}=0$.
\end{de}
Some immediate properties of basic functions are:\\

\noindent $(i)$ A linear combination of basic vector fields, with basic coefficient functions, is basic.\\

\noindent $(ii)$ The horizontal part of the commutator of two basic vector fields is basic.\\

\noindent $(iii)$ If $\w{X}$ is a basic vector field and $f$ a basic function on $M$, then $\w{X}(f)$ is a basic function on $M$.\\

I now introduce rescaled acceleration variables $\U_\alpha= \dot{u}_\alpha/(\lambda p^\prime)$ and a set of functions $\mathcal{S}= \{\O, N_{\alpha \alpha}, \b_\alpha, \Q_\alpha, \R_\alpha,\J,\E_0,\E_3, \E_{AB}\}$ ($\alpha=1,2,3$ and $A,B=1,2$), defined by

\begin{eqnarray}
& \O = \frac{p+\mu}{\lambda^{5}}\,\omega, \label{convert_O}\\
& N_{\alpha \alpha} = \frac{n_{\alpha \alpha}}{\lambda},\\
& \b_1 = -\tfrac{4}{3}\,{\frac {p+\mu}{{\lambda}^{6}}}z_1 - \tfrac{2}{3}\,{\frac
{\left(9p^\prime - 1 \right) (p+\mu)}{{\lambda}^
{5}}}\omega \U_2, \label{convert_b1}\\
& \b_2 = -\tfrac{4}{3}\,\frac {p+\mu}{\lambda^6}z_2 +
\tfrac{2}{3}\,{\frac {\left(9p^\prime - 1 \right) (p+\mu)}{{\lambda}^{5}}} \omega \U_1, \label{convert_b2}\\
& \b_3=-\tfrac{4}{3}\,{\frac {p+\mu}{{\lambda}^{6}}}z_3, \label{convert_b3}\\
& \Q_1 = -\tfrac{1}{3}\,\U_1+ \frac {q_1}{\lambda}, \quad
\R_1=\tfrac{1}{3}\,\U_1 + \frac {r_1}{\lambda}, \label{convert_Q1R1}\\
&\Q_2 = -\tfrac{1}{3}\,\U_2 + \frac{q_2}{\lambda},\quad
\R_2=\tfrac{1}{3}\,\U_2+\frac {r_2}{\lambda}, \label{convert_Q2R2}\\
&{\Q_3}+{\R_3}=-\tfrac{1}{3}\,\U_3 +\frac {q_3}{\lambda}, \quad
{\Q_3}-{\R_3}=\tfrac{1}{3}\,\U_3+\frac{r_3}{\lambda}, \label{convert_Q3R3}\\
& \J=(1-2 G)(\U_1^2+\U_2^2+\U_3^2)+\lambda^{-2}\left(\theta^2-3 \mu-2\frac{j}{p'}\right)+9\O^2\frac{\lambda^8}{(p+\mu)^2}, \label{convert_J}\\
& \E_{\alpha \beta} = \frac{3p' + 1}{\lambda^2 p'}E_{\alpha \beta} + G \U_\alpha\U_\beta, \quad (\alpha,\beta)=(1,2),\,
(1,3),\, (2,3), \label{convert_E123}\\
& \E_0 = \frac{3p' + 1}{\lambda^2 p'}(E_{11}-E_{22}) + G (\U_1^2-\U_2^2) ,\label{convert_E0}\\
& \E_3 = \frac{3p' + 1}{\lambda^2 p'}E_{33}-\frac{G}{3} (\U_1^2+\U_2^2-2\U_3^2)+\frac{2(9p'+1)\O^2}{3p'}\frac{\lambda^8}{(p+\mu)^2} .\label{convert_E3}
\end{eqnarray}

With the exception of $N_{\alpha \alpha}$ these are identical with the definitions in \cite{VdB-Radu2016}, where it was shown that $\mathcal{S}$ was a basic set (although this was proven in \cite{VdB-Radu2016} with the gauge choice $n_{11}-n_{22}=0$, the transformation properties of the right hand sides of the above expressions show that this remains valid after applying an arbitrary basic rotation). The same argument shows that also the $N_{\alpha \alpha}$ are basic (which, alternatively, can be seen to be an immediate consequence of the evolution equations for $n_{\alpha \alpha}$).

It will be convenient also to rewrite the spatial basis in terms of the basic vector fields
\begin{equation}\label{basxyz}
\X=\lambda^{-1}\partial_1, \quad \Y = \lambda^{-1}\partial_2, \quad \Z = \lambda^{-1}\partial_3.
\end{equation}
It follows then that
$\U_1= -3 \X(\ln \lambda)$, $\U_2= -3 \Y(\ln \lambda)$, and $\U_3= -3 \Z(\ln \lambda)$.

Translating the equations of Appendix 1 of \cite{VdB-Radu2016} in terms of the basic set $\mathcal{S}$
and the non-basic variables $p, \mu, \theta, \U_{\alpha}$, augmented with all information obtainable by acting with the $\partial_\alpha$ operators
on the remaining variables of the system, a set of equations was obtained which could be split
into
\begin{itemize}
\item evolution equations for the non-basic quantities $\mu, \theta, \U_\alpha$,

\item purely basic equations, presented in Appendix 2 of \cite{VdB-Radu2016},

\item algebraic equations in $\dot{U}_\alpha$ and $\theta$.
\end{itemize}
Starting from these equations a number of  theorems was derived in \cite{VdB-Radu2016}, which I review below.

\section{Recapitulation of some general theorems}\label{section_recaptheorems}

\begin{te} \label{th1}\cite{VdB-Radu2016}
If for a shear-free perfect fluid obeying a barotropic EOS $\U_1$ and $\U_2$ are basic,
then $\omega \theta = 0$.
\end{te}

A special case is of course $\U_1=\U_2=0$, i.e.~the well known situation in which the acceleration and vorticity vectors are parallel, see  \cite{WhiteCollins}, or \cite{SenSopSze,Sikhonde_Dunsby,Sikhonde_thesis} for a covariant proof.

Recently this was generalized to the following

\begin{te}\label{Muzi_th1}\cite{Sikhonde_thesis}
If for a shear-free perfect fluid obeying a barotropic EOS $\U_1$ and $\U_2$ are of the form,
\begin{equation}
\U_1=f \mathfrak{B}_1, \ \U_2=f \mathfrak{B}_2,
\end{equation}
with $f$ a function of the matter density $\mu$ and $\mathfrak{B}_1,\mathfrak{B}_2$ basic functions,
then $\omega \theta = 0$.
\end{te}

In section \ref{section_Muzitheorem1} a shorter proof of this result will be provided, once the \emph{basic SO2 formalism} has been set up.

\begin{te} \label{th2}\cite{VdB-Radu2016}
If for a rotating and expanding shear-free perfect fluid, obeying a barotropic EOS, $\U_3$ is basic, then a Killing vector along the vorticity exists.
\end{te}

As this theorem applies in particular to the case $\U_3=0$ ($\dot{\w{u}}$ orthogonal to $\w{\omega}$) and as, vice versa, the existence
of a Killing vector along the vorticity automatically implies $\U_3=0$, one also obtains
\begin{co}\label{co1}
For a rotating and expanding shear-free perfect fluid with a barotropic EOS, the acceleration is orthogonal to the vorticity if and only if a Killing vector exists along the vorticity.
\end{co}

In \cite{Sikhonde_thesis} an attempt was made to generalize this to the case where, in stead of $\U_3$ being basic, one has $\U_3=f \mathfrak{B}$ with $f=f(\mu)$ and $\mathfrak{B}$ basic. However the proof of this subcase was not completed.

\begin{te} \label{th3}\cite{VdB-Radu2016}
If a shear-free perfect fluid, obeying a barotropic EOS, admits a Killing vector along the vorticity and if the basic variables
$\b_1=\b_2=\Q_1=\R_2=\E_{12}=\E_0=\partial_1 \J=\partial_2 \J$ are 0, then $\omega \theta = 0$.
\end{te}

\begin{te} \label{th4}\cite{VdB-Radu2016}
If a shear-free perfect fluid obeys a \emph{linear} EOS (with $p+\mu\neq0$) then $\omega \theta = 0$.
\end{te}

\section{Intermezzo: description in terms of coordinates}

From the definition of $\lambda$ in (\ref{EOSconventions}), the conservation laws $\partial_0(\mu)+ \epsilon \theta=0$, $\partial_3(p)+\epsilon \u_3=0$ (with $\epsilon=\mu+p$) and the commutator relation (\ref{comm_explicit}) with $\alpha=3$, it follows that the vector fields $\epsilon^{-1}\lambda^3 \partial_0$ and $\lambda^{-1}\partial_3$ commute. This fact (which seems to have gone unnoticed before) implies the existence of coordinates $t,x,y,z$ such that
\begin{eqnarray}
\partial_0 &=& \epsilon \lambda^{-3}\frac{\partial}{\partial t},\\
\partial_3 &=& \lambda h^{-1} \frac{\partial}{\partial z},\\
\partial_I &=& \lambda (\mathcal{A}_I\frac{\partial}{\partial t}+\mathcal{B}_I\frac{\partial}{\partial z}+\mathcal{M}_I\frac{\partial}{\partial x}+\mathcal{N}_I\frac{\partial}{\partial y})\ (I=1,2), \\
\end{eqnarray}
with $ h= h(x,y,z)\neq 0$ a disposable basic function \footnote{One can put $ h=1$, but this may not be the best choice} and $\mathcal{A}_I, \mathcal{B}_I,\mathcal{M}_I,\mathcal{N}_I$ arbitrary functions.
The dual basis is then given by
\begin{eqnarray}
\w{\omega}^0 &=& \lambda^3\epsilon^{-1}\left( \ud t - \frac{\mathcal{A}_1\mathcal{N}_2-\mathcal{A}_2\mathcal{N}_1}{\Delta}\ud x + \frac{\mathcal{A}_1\mathcal{M}_2-\mathcal{A}_2\mathcal{M}_1}{\Delta}\ud y \right),\\
\w{\omega}^3 &=& \lambda^{-1}  h \left(\ud z - \frac{\mathcal{B}_1\mathcal{N}_2-\mathcal{B}_2\mathcal{N}_1}{\Delta}\ud x + \frac{\mathcal{B}_1\mathcal{M}_2-\mathcal{B}_2\mathcal{M}_1}{\Delta}\ud y \right),\\
\w{\omega}^1 &=& \lambda^{-1} \Delta^{-1} [ \mathcal{N}_2\, \ud x -\mathcal{M}_2\, \ud y],\\
\w{\omega}^2 &=& \lambda^{-1} \Delta^{-1} [ -\mathcal{N}_1\, \ud x +\mathcal{M}_1\,  \ud y],
\end{eqnarray}
with $\Delta=\mathcal{M}_1\mathcal{N}_2-\mathcal{M}_2\mathcal{N}_1$.

Inserting this into the remaining commutator relations (\ref{comm_explicit}), it follows that $\mathcal{A}_I, \mathcal{B}_I,\mathcal{M}_I,\mathcal{N}_I$ are \emph{basic} functions. \footnote{Additionally the $[\partial_I,\, \partial_3]$ commutator relations imply
\begin{equation*}
\frac{\partial}{\partial z}\left(\frac{\mathcal{A}_1\mathcal{M}_2-\mathcal{A}_2\mathcal{M}_1}{\Delta}\right)=\frac{\partial}{\partial z}\left(\frac{\mathcal{A}_1\mathcal{N}_2-\mathcal{A}_2\mathcal{N}_1}{\Delta}\right)=0,
\end{equation*}
showing that a transformation $t\to t +f(x,y)$ exists such that $\mathcal{A}_1\mathcal{N}_2-\mathcal{A}_2\mathcal{N}_1=0$ or $\mathcal{A}_1\mathcal{M}_2-\mathcal{A}_2\mathcal{M}_1=0$.}

A significant simplification occurs when a Killing vector exists along the vorticity. In section \ref{sectionKV} I present a simple argument showing that in this case $n_{11}-n_{22}=n_{12}=n_{33}=0$, in addition to the obvious conditions $\partial_3(\theta)=\partial_3(\mu)=0$. Fixing the tetrad such that $\w{e}_3$ is identically 0 when acting on invariantly defined objects guarantees by the (13) and (23) field equations that also $n_{11}=n_{22}=0$ (see Remark 5 of \cite{VdB-Radu2016} for an argument showing that such a fixation is always possible). Substituting all this in the commutator relations (\ref{comm_explicit_bis}) one finds that $\mathcal{A}_I,\mathcal{M}_I,\mathcal{N}_I$ are independent of $z$, implying the existence of coordinates $x,y$ such that $\mathcal{M}_2=\mathcal{N}_1=0$ and $\mathcal{M}_1=\mathcal{N}_2:=g(x,y)$. By a $(x,y)$ dependent translation of $t$ and a transformation of the $z$-coordinate one achieves $\mathcal{A}_2=\mathcal{B}_2=0$ (herewith $ h$ is no longer disposable). From (\ref{comm_explicit_bis}) it follows that $h$ and $\mathcal{B}_1 / g$ are independent of $y$, such that a further transformation of the $z$ coordinate makes $\mathcal{B}_1=0$, while $h$ becomes separable in $x$ and $z$. Absorbing the $z$-dependent factor in $\ud z$, finally leads to the following form of the tetrad:
\begin{eqnarray}
\partial_0 &=& \epsilon \lambda^{-3}\frac{\partial}{\partial t}, \label{ds_simpleforma}\\
\partial_3 &=& \lambda h^{-1} \frac{\partial}{\partial z},\\
\partial_1 &=& \lambda (f\frac{\partial}{\partial t} + g \frac{\partial}{\partial x}), \\
\partial_2 &=& \lambda g \frac{\partial}{\partial y},\label{ds_simpleformd}
\end{eqnarray}
with $f=f(x,y), g=g(x,y)$ and $ h= h(x)$.

The simple form of the corresponding line-element is misleading: expressing that the right hand-side of the field equations is given by the energy-momentum tensor (\ref{intro_Tab}) with $p=p(\mu)$ remains a difficult task. Even when $ h=1$, corresponding to the special case $\kappa=0$ studied in section \ref{sectionKV}, the `pure tetrad approach' followed in the latter section appears to be more efficient than the coordinate approach.

\section{Complexification}\label{complexification}
It is no exaggeration to say that the equations in \cite{VdB-Radu2016} look quite complicated. In order to bring some structure into these, one first notices that the appearance of the equations begs for the introduction of complex variables, making the transformation properties under basic rotations in the 2-plane orthogonal to the vorticity explicit. Therefore, as a first step, I introduce a complex acceleration variable $\u$ and a differential operator $\d$ by
\begin{eqnarray}
\u &=& \U_1+ i \U_2= \frac{\u_1+i \u_2}{\lambda p'}, \label{ud_def1}\\
\d  &=& X + i Y = \frac{1}{\lambda} (\partial_1+i \partial_2). \label{ud_def2}
\end{eqnarray}
Guided by the transformation properties of the variable pairs $(\Q_1,\R_2)$, $(\Q_2,\R_1)$, $(\Q_3, N_{11}-N_{22}), (\E_{13},\E_{23}), (\E_0, 2 \E_{12})$ of \cite{VdB-Radu2016} I also define new quantities $\kappa,\tau,\pi,\rho,\Eone,\Etwo$ (complex ) and $\nu,\Ethree$ (real) by
\begin{eqnarray}
\kappa &=& \Q_1- i \R_2,\label{complex_def1}\\
\pi &=& \Q_3-\frac{i}{2} (N_{11}-N_{22})\label{complex_def2},\\
\tau &=& \Q_2+i \R_1,\label{complex_def3}\\
\rho &=& N_{33}+ 2 i \R_3,\label{complex_def4}\\
\nu &=& N_{33}-N_{11}-N_{22}, \label{complex_def5}\\
\Eone &=& \E_{13}+ i \E_{23},\label{complex_def6}\\
\Etwo &=& \tfrac{1}{3} (\E_0+2 i \E_{12}),\label{complex_def7}\\
\Ethree &=& \E_3 . \label{complex_def8}
\end{eqnarray}
From the definitions (\ref{convert_O}, \ldots , \ref{convert_Q3R3}), (\ref{ud_def1}, \ldots ,\ref{complex_def8}) and the commutator relations for the $\partial_a$ vector fields some straightforward algebra leads then to the following commutator relations involving the new derivative operators:
\begin{eqnarray}
 {}[ \partial_0,\, \d ] &=& p' \u\, \partial_0, \label{defcom0d} \\
 {}[ \partial_0,\, \Z ] &=& p' \U_3\, \partial_0, \label{defcom0Z}\\
 {}[ \Z,\, \d ] &=& \frac{i}{2}(\nu-\overline{\rho})\, \d +\kappa \, \Z +\pi \, \dc, \label{defcomZd}\\
 {}[ \dc,\, \d ] &=& -4 i \frac{\O \lambda^3}{p+\mu}\, \partial_0 +i \overline{\tau}\, \d + i \tau \, \dc +i (\rho+\overline{\rho})\, \Z ,\label{defcomdcd}
\end{eqnarray}
providing an independent check that $\kappa,\pi,\tau,\nu,\rho$ are basic.

In what follows I will often be using, instead of the variable $\epsilon = \mu+p$ introduced in \cite{VdB-Radu2016}, the variable
\begin{equation}\label{defeps0}
\ez = \frac{\epsilon}{\lambda^4},
\end{equation}
while the propagation of $\theta$ along the fluid flow will be written as
\begin{equation}\label{Thetadef}
\Theta = \partial_0 \theta
\end{equation}
and the projection of the gradient of $\theta$  in the (12)-plane (in accordance with (\ref{convert_b1},\ref{convert_b2})) as\footnote{$z$ here is not to be confounded with the $z$-coordinate introduced in the previous section.}
\begin{equation}
\d \theta = z.
\end{equation}

Translating the results of \cite{VdB-Radu2016}, using the above complex variables, one finds then
\begin{enumerate}
\item the derivatives of the acceleration and the expansion (see Appendix 1).
\item a set of basic equations (see Appendix 2),
\item an evolution equation for $z$, which together with (\ref{e0u}) implies that the following combination of $u$ and $z$ is basic:
\begin{equation}
(\frac{1}{3} - 3 p')\O \u-\frac{2 i}{3 \lambda}\ez z = \beta \label{z_u_relation} ,
\end{equation}
where the basic function $\beta$ (which is an `integration constant') is related to $\b_1,\b_2$ by (\ref{convert_b1},\ref{convert_b2}):
\begin{equation}
\b_1+i\b_2 = -2 i \beta
\end{equation}
while $\b_3$ of equation (\ref{convert_b3}) turns out to be related to $\rho$ by $\b_3= -\O (\rho+\rhoc)$. Herewith one also obtains expressions for the spatial gradient of $\O$ (see Appendix 2).
\end{enumerate}

\textbf{Remark 1}

In order to obtain the complex equations in a more direct way, one can start by \emph{assuming}\footnote{rather than by painstakingly working through the procedure of \cite{VdB-Radu2016} to demonstrate that they are basic} the quantities $\kappa,\tau,\pi,\rho,\nu$ to be basic, after which the $(0\alpha)$ Einstein equations and the
$\partial_0 ( a_\alpha )$ Jacobi equations result in expressions (\ref{ap0u},\ref{ap0U3},\ref{apZtheta}, \ref{apdO}), with (\ref{apdO}) showing that the lhs of (\ref{z_u_relation}) is basic indeed. Next, defining $\J,\Eone,\Etwo,\Ethree$ by (\ref{convert_J}-\ref{convert_E3},\ref{complex_def6}-\ref{complex_def8}), equations (\ref{apdu}-\ref{apZu},\ref{apdU3},\ref{apZU3}) follow from the expressions for the electric part of the Weyl tensor \footnote{see for example \cite{Norbert2013}}, while (\ref{apdO}-\ref{Beq1}) follow from the spatial part of the Einstein equations and the `$\mathbf{\nabla \cdot n}$' Jacobi equations. The operators $\d,\dc,\Z$ being basic, this last set of equations confirms then also that $\J,\Eone,\Etwo,\Ethree$ are basic.

\textbf{Remark 2}

For the tetrad (\ref{ds_simpleforma}-\ref{ds_simpleformd}), describing the situation of a shear-free perfect fluid with a Killing vector aligned with the vorticity, the only surviving basic rotation coefficients are $\kappa, \tau$ and $\O$, related to the functions $f,g, h$ by
\begin{equation}
\kappa = g \frac{ h_{,x}}{ h},\ \tau = i( g_{,x}+i g_{,y}),\ \Omega =\frac{1}{2}(g f_{,y}-f g_{,y})
\end{equation}
and with the basic function $\beta$ (defined by (\ref{z_u_relation})), involving also the derivatives of $\mu$ w.r.t.~$x,y$ and $t$.

\section{Basic SO2 formalism}\label{section_basic_SO2}
In analogy with \cite{GHP} or \cite{Wylleman_thesis} a quantity $\Phi$ will be said to be \emph{well weighted} if it transforms under a basic rotation (\ref{basicrotations}) as
\begin{equation}
\Phi \to e^{i W \alpha} \Phi \ (W\in \mathbb{Z})
\end{equation}
and $W$ will be called the weight of $\Phi$: $W=w(\Phi)$. From their definitions it is clear that all quantities defined in section \ref{complexification}, with the exception of $\nu$ and $\tau$, are well weighted, with weights $0,\pm 1,\pm 2$:
\begin{itemize}
  \item 0: $\mu, p, \ez, \lambda, G, \theta, \U_3, \rho, \O, \J, \Ethree$ and the operators $\w{e}_0$ and $\Z$
  \item 1: $\u, \kappa, \beta, \Eone$ and the operator $\d$
  \item 2: $\pi, \Etwo$
\end{itemize}
The fact that well weighted basic objects don't remain well weighted under the action of $\d$ or $\Z$ can be remedied by `absorbing' the badly transforming rotation coefficients $\nu$ and $\tau$ in $\d$ and $\Z$: define therefore new derivative operators $\et,\etd,\ZZ$ by
\begin{eqnarray}
\et(x) = \d(x) -i w(x) \tau x,\\
\etd(x) = \dc(x) - i w(x) \tauc x,\\\label{def_etd}
\ZZ(x) = Z(x) -\frac{i}{2} w(x) \nu x ,
\end{eqnarray}
where $x$ is an arbitrary well weighted function on space-time and with the sign in the rhs of (\ref{def_etd}) being due to the fact that $w(\overline{x})=-w(x)$ and $\overline{\et(x)}=\etd(\overline{x})$.
It is then clear that $\et, \etd$ and $\ZZ$ map objects of weight $W$ to objects of weight $W+1, W-1, W$ respectively:
 \begin{equation}
 w(\et)=1, \ \ w(\etd) = -1, \ \ w(\ZZ) = 0 .
 \end{equation}
Herewith $\nu$ and $\tau$ disappear from the formalism, with the main equations now simplifying to
\begin{enumerate}
\item equations for matter density, acceleration and expansion:

 \begin{eqnarray}
 \partial_0(\mu) &=& -(\mu+p) \theta, \label{e0mu}\\
 \et(\mu) &=& -(\mu+p) \u, \label{etmu} \\
 \ZZ(\mu) &=& -(\mu+p) \U_3, \label{ZZmu}
 \end{eqnarray}

\begin{eqnarray}
 \partial_0(\u) &=& p' \u \theta +z,\label{e0u} \\
 \et (\u)  &=& (\frac{3 p'G}{1+3 p'} +\tfrac{1}{3}) \u^2  + 2 \pi  \U_3 + \frac{3}{1 + 3 p'}\Etwo , \label{etu}\\
 \etd (\u) &=& \frac{3 p'G}{1+3 p'} \modU - ( \frac{2 p'G}{1+3 p'}+\tfrac{1}{3}) \U_3^2 + i \rho \U_3 + \frac{\theta^2}{3 \lambda^2}  \nonumber \\
  & & - 2 i \frac{\O\theta}{\lambda \ez} + \frac{9 p' + 7}{\ez^2 (1 + 3 p')} \O^2 - \frac{\mu}{\lambda^2} - \frac{\J}{3} - \frac{\Ethree}{1 + 3 p'} ,\label{etdu}\\
 \ZZ(\u)  &=& (\frac{3 p'G}{1+3 p'}+\tfrac{1}{3}) \U_3 \u +\kappa \U_3+\frac{3}{1+3 p'} \Eone,\label{ZZu}
\end{eqnarray}

\begin{eqnarray}
\partial_0(\U_3) &=& p' \U_3 \theta + \frac{3 \lambda}{4 \epsilon_0} \O (\rho+\rhoc),\label{e0U3}\\
\et (\U_3) &=& [ (\frac{3 p'G}{1+3 p'}+\tfrac{1}{3}) \U_3 + i \frac{\rhoc}{2} ] \u- \pi \uc +\frac{3}{1+3p'} \Eone
,\label{etU3}\\
\ZZ (\U_3) &=& -( \frac{ p'G}{1+3 p'} +\tfrac{1}{6}) \modU - \tfrac{1}{2}(\kappac \u+ \kappa \uc)
+ (\frac{2 p'G}{1+3 p'}+\tfrac{1}{6}) \U_3^2 \nonumber \\
& & + \frac{\theta^2}{6 \lambda^2} + \frac{9 p' - 5}{2 \ez^2 (1 + 3 p')} \O^2 - \frac{\mu}{2 \lambda^2} - \frac{\J}{6} + \frac{1}{1 + 3 p'} \Ethree,\label{ZU3}
\end{eqnarray}

\begin{eqnarray}
\partial_0(\theta) &=& \frac{\lambda^2 p'}{2}(1-2 G)(\modU +\U_3^2) +(\frac{p'}{2}-\frac{1}{3})\theta^2+\frac{\lambda^2}{2 \ez^2}(9p'+4)\O^2 \nonumber \\
& & -\frac{\lambda^2 p'}{2} \J -\frac{3}{2} \lambda^4\ez +(1-\frac{3}{2}p')\mu  , \label{Raych}\\
\et (\theta) &=& z , \label{ettheta} \\
\ZZ (\theta) &=& \frac{3 \lambda}{4\ez } \O (\rho+\rhoc), \label{ZZtheta}
\end{eqnarray}
with $z$ related to $\u$ by (\ref{z_u_relation}) and with $\modU$ an auxiliary variable defined by

\begin{eqnarray}
\modU &=& \u \uc .
\end{eqnarray}

\item basic equations:
\begin{eqnarray}
\et (\O) = \beta -\kappa \O, \label{dOmega}\\
\ZZ (\O) = \frac{i}{2} \O (\rho - \rhoc), \label{ZOmega} \\
\et (\rho)+2 i \etd (\pi) = -(\rho+\rhoc ) \kappa-2 i \Eone , \label{beq5}\\
\et (\kappa)-2 \ZZ(\pi)= i (\rhoc -\rho)\pi-\kappa^2-\Etwo \label{beq3} ,\\
\etd (\kappa)-i \ZZ(\rho) = -|\kappa |^2-2 |\pi|^2 +\tfrac{2}{9} \J-\tfrac{1}{3} \Ethree+\tfrac{1}{2} \rho^2 , \label{beq4}
\end{eqnarray}

\end{enumerate}

together with their complex conjugates.\\

This set has to be augmented by the new commutator relations, namely

\begin{eqnarray}
 {}[ \partial_0,\, \et ](x) &=& p' \u\, \partial_0(x), \label{Defcom0d} \\
 {}[ \partial_0,\, \ZZ ](x) &=& p' \U_3\, \partial_0(x), \label{Defcom0Z}\\
 {}[ \ZZ,\, \et ](x) &=& -\frac{i}{2}\overline{\rho}\, \et(x)  +\pi \, \etd(x) +\kappa \, \ZZ(x) \nonumber \\
 & & -(\frac{i}{2}\rhoc \kappa+\pi\kappac-\Eone)x\, w(x), \label{DefcomZd}\\
 {}[ \etd,\, \et ](x) &=& -4 i \frac{\O}{ \lambda \ez}\, \partial_0(x) +i (\rho+\overline{\rho})\, \ZZ(x)\nonumber \\
  & & +(\frac{1}{2}|\rho|^2-2|\pi|^2-\frac{2}{9}\J-\frac{2}{3}\Ethree) x \, w(x),\label{Defcomdcd}
\end{eqnarray}

With the aid of the commutator relations it is straightforward, although cumbersome, to obtain the integrability conditions for the equations (\ref{e0u}-\ref{ZU3}). The resulting \emph{algebraic} relations between the main unknowns $\u, \uc, \U_3, \theta$, the variables which depend on $\mu$ only and the basic variables, are presented in Appendix 3.

Using the integrability conditions of Appendix 3 one can eliminate $\beta$ from the derivatives of (\ref{z_u_relation}) and thereby obtain the following (sometimes useful) expressions for the derivatives of $z$\footnote{Alternatively the $z$-derivatives can be obtained directly from the Einstein field equations and Jacobi equations: this procedure was followed in Appendix A of \cite{VdB-Radu2016} in order to show that the lhs of (\ref{z_u_relation}) is indeed basic and thereby obtain (\ref{eq5},\ref{eq10},\ref{eq6})}.

\begin{eqnarray}
\partial_0 (z) &=& \half i p'(2-9 G) \frac{\lambda}{\ez} \O \theta \u + [(p'-\tfrac{2}{3})\theta+\half i (9p'-1)\frac{\lambda}{\ez} \O] z \label{e0z}
\end{eqnarray}

\begin{eqnarray}
\et (z) &=& \frac{3 p'}{(1+3 p')^2}[-\epsilon G' (1+3 p') - G^2 +\tfrac{2G}{3}  (9 {p'}^2+3 p' +1) ] \theta \u^2 \nonumber \\
& & +(6 \frac{p'}{1+3 p'} G -2 p' +\tfrac{2}{3}) z \u  +9 \Etwo \frac{ p' }{(1+3 p')^2}(G -2) \theta \nonumber \\
& & +\tfrac{3}{2}(\rho+\rhoc) \O \pi\frac{\lambda}{\ez} \label{etz}
\end{eqnarray}

\begin{eqnarray}
\etd (z) &=& \frac{p'}{1+3 p' }( G - 1- 3 p' )( \zc \u+ z \uc) +i \frac{\lambda {p'}}{\ez} (2 G-1)  \O  (\modU+\U_3^2) \nonumber \\
& & -\frac{p'}{(1+3 p' )^2}[ (1+3 p') \epsilon G' + G^2 -\tfrac{2}{3}(9 {p'}^2+3 p' +1) G ] \theta (\modU-2 \U_3^2) \nonumber \\
& & -(\rho+\rhoc) \frac{\lambda}{\ez} (3 \frac{p' G}{1+3 p'}  +\half) \O \U_3-\tfrac{1}{3} i \frac{3 p'-2}{\lambda \ez} \O \theta^2 \nonumber \\
& &+\frac{2}{\ez^2 (1+3p')^2}[6 p'G  +\tfrac{1}{3} (81 {p'}^3+81 {p'}^2-9 p' -5)] \O ^2 \theta \nonumber\\
& & - \Ethree \frac{p'(3 G -2)}{(1+3 p' )^2} \theta-i \frac{(9 p' +4) \lambda}{\ez^3} \O^3 \nonumber \\
& & +i \frac{\O}{4\ez \lambda} [12 \ez  \lambda^4 +4 \lambda^2 p' \J+3 (\rho +\rhoc) \rho \lambda^2 +4 \mu (3 p' -2)]\label{etdz}
\end{eqnarray}

\begin{eqnarray}
\ZZ(z) &=& \frac{3 p'}{(1+3 p')^2}[-\epsilon G' (1+3 p') - G^2 +\tfrac{2G}{3}  (9 {p'}^2+3 p' +1) ] \theta \U_3 \u \nonumber \\
& & + (3 \frac{p'}{1+3 p'} G -p' +\tfrac{1}{3})(\frac{3 \lambda}{4 \ez} \O (\rho+\rhoc ) \u+ z \U_3) \nonumber \\
& & +\Eone\frac{3 p'}{(1+3 p')^2} (3 G -2 ) \theta+\frac{3\lambda \O}{4\ez}  (\rho+\rhoc )\kappa \label{ZZz}
\end{eqnarray}

\section{A useful lemma}\label{sectionlemma}

When trying to prove the conjecture one often encounters a situation where some  $\mu$-dependent linear combination of basic quantities vanishes: $\Sigma _{i=1}^n a_i(\mu) b_i = 0$,
with $\partial_0(b_i)=0$ ($i=1 \ldots n$). From the equations (\ref{e0mu},\ref{etmu}) useful information can then be extracted:

\begin{lm}\label{lincolemma}
If $a_i \neq 0$ are differentiable functions of $\mu$ and $b_i \neq 0$ are basic functions ($i=1 \ldots n$) with $\theta\neq 0$, then $\Sigma _{i=1}^n a_i  b_i = 0 \Rightarrow $
\begin{equation*} \exists r \  0 \leq r < n: \ \textrm{c-rank}\,(a_1, \ldots a_n)=r \textrm{ and  c-rank}\,(b_1, \ldots b_n)=n-r,
\end{equation*}
\end{lm}
\noindent where a set of functions $f_i$ is defined to be of c-rank $r$ if a permutation $\pi$ of $(1\ldots n)$ exists such that $f_{\pi(1)} \ldots f_{\pi(r)}$ are linearly independent and
$f_{\pi(r+1)} \ldots f_{\pi(n)}$ are \emph{constant} linear combinations of $f_{\pi(1)} \ldots f_{\pi(r)}$. When all $f_i$ are constant, we say that $\textrm{c-rank}\, (f_1, \ldots f_n)=0$.\\
This is nothing but a `separation of varables' property, which is trivially true for $n=2$ ($a_1 b_1 + a_2 b_2 = 0$ implying $\w{e}_0(a_1 /a_2) =0$ and hence c-rank$(a_1,\,a_2)=1$) and which holds by induction for all other values of $n$.

\section{Application 1}\label{section_CWtheorem1}

As a first application of the formalism, as well as for completeness\footnote{The `parallel case' deserves particular attention, as it has been relied on in many subsequent proofs.}, I present a simple proof of the conjecture for the case where
the vorticity and acceleration vectors are parallel to each other (see \cite{WhiteCollins} for the first tetrad-based proof, or \cite{SenSopSze,Sikhonde_Dunsby,Sikhonde_thesis} for a covariant proof):


\begin{proof}
When the vorticity and acceleration vectors are parallel to each other, we have $\u=0$ and hence\footnote{throughout we also assume $p'\neq0\neq 3 p'+1$: see the introduction for these classical cases}, by (\ref{e0u}-\ref{ZZu}),
\begin{eqnarray}
\fl z = 0,\  2\U_3\pi  + \frac{3}{1+3 p'}\Etwo = 0,\  \U_3\kappa  + \frac{3}{1+3 p'}\Eone = 0, \label{paralkeyabc} \\
\fl \frac{\theta^2}{3 \lambda^2} -  \frac{6 G + 3 p' + 1}{3(1+3 p')} \U_3^2 +i\frac{\rho-\rhoc}{2}\U_3+ \frac{9 p' + 7}{\ez^2 (1 + 3 p')} \O^2
 - \frac{\mu}{\lambda^2} - \frac{\J}{3} - \frac{\Ethree}{1 + 3 p'} = 0 , \label{paralkeyd} \\
\fl \theta = \U_3 \lambda \ez \frac{\rho+\rhoc}{4 \O}. \label{paralkey1}
\end{eqnarray}

It follows that $(\rho+\rhoc )\U_3 \neq0$, as otherwise $\mu$ or $\theta$ would be non-constant functions\footnote{Note that $\theta=constant$ would imply, by (\ref{apZtheta}), $\rho+\rhoc=0$ and hence, by (\ref{paralkey1}), $\theta=0$.} with a vanishing spatial gradient and hence the vorticity would be 0. Notice also the simple relation following
from (\ref{apsol4}),
\begin{equation}\label{paralkey0}
\Theta  = \frac{3}{16} (\rho+\rhoc)^2 \lambda^2,
\end{equation}
while $\beta=0$ and $\ZZ(\rho+\rhoc)=0$ follow from (\ref{z_u_relation}) and (\ref{beq2}).

Evaluating $\ZZ(\ref{paralkey1})$ (using (\ref{Raych}) and  (\ref{paralkeyd}) to eliminate $\Ethree$ and $\J$) one obtains
\begin{eqnarray}
\fl \ez^2 (3 G-3 p'-2) p' \lambda^2 \U_3^4+[\ez^2 (\theta^2 - 3 \mu )+\tfrac{9}{2} \lambda^4 \ez^3-3 \O^2 (3 p'+2) \lambda^2] \U_3^2+9 \O^2 \theta^2
= 0 \label{paralmain},
\end{eqnarray}
the $\et$ derivative of which (using (\ref{paralkeyabc}) to eliminate $\Eone$) gives

\begin{eqnarray}
\kappa [-\frac{3}{2}\lambda^4 \ez^3 \U_3^2+ \U_3^2(\mu -\tfrac{1}{3}\theta^2) \ez^2-3\O^2\theta^2] = 0. \label{paraldmain}
\end{eqnarray}
On the other hand $\partial_0 (\ref{paralkeyd})$ implies
\begin{eqnarray}
\kappa [(2-3 G)\ez^2 p'\U_3^2+3 \O^2 (1+3 p')] = 0, \label{paralkeyd0}
\end{eqnarray}
which, together with (\ref{paraldmain}), shows that $\kappa=0$, as otherwise elimination of $G$ and $\theta$ from (\ref{paralmain},\ref{paraldmain},\ref{paralkeyd0}) would imply $\U_3^2 ({p'}^2\U_3^2\ez^2 + \O^2)=0$ and hence $\U_3=0$.

With $\kappa=0$ it follows that $\rho$ cannot be real: otherwise (\ref{paralkey1}) would imply $\theta = \U_3 \lambda \ez \frac{\rho}{2 \O}$, which, after eliminating $G, G'$ from $\partial_0(\ref{paralmain})$ and (\ref{etdz}), would lead to $\lambda^2 \ez \theta=0$.

This however is in contradiction with (\ref{eq8}), which (eliminating the $\Ethree$ and $\J$ terms as before) gives
\begin{eqnarray}
3\lambda \ez (\pic \Etwo - \pi \Etwoc) - 2\O (1 + 3p')(\rhoc -\rho) \theta = 0,
\end{eqnarray}
with $\pic \Etwo - \pi \Etwoc$ being 0 as a consequence of (\ref{paralkeyabc})[2].

\end{proof}

\section{Application 2}\label{section_Muzitheorem1}

As a second application of the formalism I give an alternative proof of Theorem \ref{Muzi_th1}. This generalizes the `parallel case' of the previous section to the case where the complex acceleration is of the form $\u = f \mathfrak{B}$ with $f = f(\mu)$ real and $\mathfrak{B}$ (complex) basic\footnote{This was first demonstrated for the special case $f(\mu)=\lambda p'$ in \cite{VdB-Radu2016}, with a covariant proof for arbitrary $f(\mu)$) being presented in \cite{Sikhonde_thesis}).}.
\begin{proof}
If $\u = f \mathfrak{B}$ with $f=f(\mu)$ real and $\mathfrak{B}$ (complex) basic (hence $w(\B)=1$), then $\partial_0(\frac{\u}{f})=0$ implies
\begin{equation}
 F\theta \u + z = 0 \textrm{ with } F =-\frac{f'}{f}\lambda^4\ez + p' \label{app1_1}.
\end{equation}
The orthogonal projections (in the 2-plane orthogonal to the vorticity) of the spatial gradients of $\mu$ and $\log(\theta)$ being aligned, with a coefficient ($F$) which is a function of $\mu$ only, suggests one to define a function $S$ by
\begin{equation}
S=\log(\theta) -\int \frac{F}{\mu+p} \ud \mu,
\end{equation}
implying
\begin{eqnarray}
\partial_0 (S) = \frac{\Theta}{\theta}+\theta F,\\
\et (S) = \etd (S) = 0,
\ZZ (S) = 3\frac{ \lambda \O (\rho + \overline{\rho})}{4 \theta \ez} + \U_3 F .
\end{eqnarray}
I will show that $\ZZ(S)=0$, implying that either $S$ is constant and therefore $\theta=\theta(\mu)$, whence (see \cite{Lang,Langth} or section \ref{section_Langtheorem1}) $\omega \theta = 0$, or that $S$ is a non-constant function with vanishing spatial gradient, implying $\omega=0$. First consider the special case $F=0$:\\

(i) When $F=0$ (\ref{z_u_relation}) gives $\u = -3\frac{\beta}{\O}\frac{1}{9p' - 1}$ and hence $f=(9 p'-1)^{-1}$. Evaluating herewith (\ref{e0u}) shows that $9 G -2 =0$, after which (\ref{etu},\ref{ZZu}) yield two polynomials, linear in $\U_3$ and with coefficients being functions of $p'$ and basic functions. Elimination of $\U_3$ results in a polynomial of 4th degree in $p'$, with basic coefficients, implying a linear EOS unless all coefficients are 0. Some straightforward algebra leads then to
\begin{equation}
(54\O^2 \Etwo+54 \beta^2)p' -6\Etwo\O^2+19\beta^2=0
\end{equation}
and hence $\beta$=0, such that the acceleration becomes parallel with the vorticity.\\

ii) When $F\neq 0$ (\ref{z_u_relation}) and (\ref{app1_1}) imply
\begin{equation}
\theta = -i \frac{\lambda(9p' - 1)}{2 F\ez}\O - 3 i \frac{\beta}{\mathfrak{B}} \frac{\lambda f}{2 F\ez}, \label{app1_2}
\end{equation}
the real part of which reads
\begin{equation}
\tfrac{3}{2}\frac{\beta \overline{\mathfrak{B}}+\betac \mathfrak{B}}{\O \mathfrak{B} \overline{\mathfrak{B}}} +\frac{9 p'-1}{f}=0
\end{equation}
and hence $f=9 p'-1$ (modulo a constant which is absorbed in $\mathfrak{B}$). Then $F = \frac{ p'(2-9 G)}{9 p' - 1}$ and, by (\ref{app1_2}),
\begin{equation}
\theta = \frac{i}{2} \frac{ \lambda (9 {p'} - 1)^2 }{{p'} (9 G - 2) \ez}  (\O + 3 \frac{\beta}{\mathfrak{B}}) .\label{app1_3}
\end{equation}
Propagating this along $\w{e}_0$ one finds
\begin{equation}
\Theta = \theta^2\left(  \frac{9 \epsilon G'}{2-9 G}-\frac{9 G {p'}+G+5 {p'}-1}{9 {p'}-1}\right) ,
\end{equation}
which, when substituted in the $[\partial_0,\, \et]S$ and $[\partial_0,\, \ZZ]S$ commutator relations, yields the following algebraic relations:
\begin{eqnarray}
 9 (\epsilon G')'  \lambda^4 \ez-\frac{81}{9 G-2} (\epsilon G')^2-\frac{243 G {p'}+81 {p'}^2+9 G-63 {p'}-2}{9 {p'}-1} \epsilon G'\nonumber \\
 +\frac{9 G-2}{(9 {p'}-1)^2} (3 G+3 {p'}-1) (54 G {p'}+12 G+9 {p'}-5) {p'}= 0 \label{app1_eos1} \\
 \left( (486 {p'}+54) G^2-(486 {p'}-54) \epsilon G'+(243 {p'}^2-27 {p'}-48) G \right. \nonumber \\
\left. -54 {p'}^2-18 {p'}+8 \right) \ZZ(S) = 0 .\label{app1_eos2}
\end{eqnarray}
When the first factor of (\ref{app1_eos2}) is 0, propagating it and using (\ref{app1_eos1}) to eliminate successively the $G'' $ and $G'$ terms, one obtains a linear EOS. Hence $\ZZ(S)=0$, thereby finishing the proof.
\end{proof}

\section{Application 3}\label{section_Langtheorem1}

This section presents a compact proof of the validity of the conjecture when matter density $\mu$ and expansion $\theta$ are functionally related (see \cite{Lang,Langth, Sopuerta1998} for the original proofs, as well as the remarks in the introduction). Use will be made of the property that $\u=0$ implies the validity of the conjecture\footnote{See \cite{WhiteCollins} for the first tetrad-based proof, \cite{SenSopSze,Sikhonde_Dunsby,Sikhonde_thesis} for a covariant proof.}.

\begin{proof}
We begin by expressing that $\theta=\theta(\mu)$ by the relations
\begin{equation}
\partial_0(\theta) \et(\mu)-\partial_0(\mu) \et(\theta)= \ZZ(\theta)\et(\mu)-\et(\theta) \ZZ(mu)=0,
\end{equation}
which, using (\ref{e0mu}-\ref{ZZmu}) and (\ref{Raych}-\ref{ZZtheta}) imply
\begin{eqnarray}
z = \frac{\Theta}{\theta} \u , \label{Langkey3} \\
\ez \Theta \U_3 = \tfrac{3}{4}\lambda \theta (\rho+\rhoc) \O. \label{Langkey1}
\end{eqnarray}

Substituting (\ref{Langkey3}) in (\ref{z_u_relation}) one first notices that the basic function $\beta \neq 0$. In fact, when $\beta=0$ the real and imaginary parts of (\ref{z_u_relation}) tell us that $\Theta= z = \rhoc +\rho = p' -\tfrac{1}{9}=0$ and hence, by (\ref{ZZtheta}), $\theta=constant$. Herewith the Raychaudhuri equation (\ref{Raych}) reduces to
\begin{equation}\label{LangRaych}
405 \frac{\lambda^2\O^2}{\ez^2}-243\epsilon -9\lambda^2\J+5\lambda^2(\modU+\U_3^2)-45\theta^2+135\mu = 0,
\end{equation}
which, after propagating three times along $\partial_0$ and eliminating $\modU+\U_3^2$, successively results in
\begin{eqnarray}
135\frac{\O^2}{\ez^2}-\J-\tfrac{27}{2}\lambda^2\ez+\frac{10}{\lambda^2}(\theta^2-3 \mu) = 0 ,\\
45\O^2-27\ez^3\lambda^2-5\frac{\ez^2}{\lambda^2}(\theta^2-3\mu) = 0,\\
\mu=\tfrac{9}{10}\ez\lambda^4+\tfrac{1}{3}\theta^2.
\end{eqnarray}
Solving for $\J$ and $\O$ one obtains $\J=0$ and $\O^2=\tfrac{3}{10} \ez^3 \lambda^2$, thereby reducing (\ref{LangRaych}) to $\modU+\U_3^2=0$ and hence $\O=0$ by (\ref{etdu}).

Continuing with $\beta\neq0$ it follows from (\ref{Langkey3}) and (\ref{z_u_relation}) that
\begin{equation}
\u = - 3 \beta [\O(9 p'-1) +2 i \ez \Theta (\lambda\theta)^{-1}]^{-1}  \label{Langu}
\end{equation}
with $\Theta= -\theta \theta' \epsilon$ by (\ref{Thetadef}). In what follows I will be using (i) the real and imaginary parts of (\ref{apsol2}),
\begin{eqnarray}
p'+ \tfrac{1}{3} +\epsilon \frac{\theta''}{\theta'}+\epsilon \frac{\theta'}{\theta} = 0 \label{Langkey8},\\
p'(9 G -1)+\epsilon \frac{\theta'}{\theta}(9p'-1) = 0 , \label{Langkey9}
\end{eqnarray}
together with their propagations (leading to expressions for $\theta '''$ and $G'$) and (ii) the evolution equations (\ref{e0u}, \ref{e0U3}), which under the conditions (\ref{Langkey3},\ref{Langkey1},\ref{Langu}) simplify to
\begin{eqnarray}
\partial_0(\u) =  \u (\theta p'-\epsilon \theta'),\label{Langevolsu}\\
\partial_0(\modU) = 2 \modU (\theta p'-\epsilon \theta'),\label{LangevolsmodU}\\
\partial_0(\U_3) =  \U_3 (\theta p'-\epsilon \theta') \label{LangevolsU3}.
\end{eqnarray}
We now act with the $\ZZ$ operator on (\ref{Langkey1}), using (\ref{beq2}) and (\ref{beq6}) to get rid of the $\ZZ(\rho+\rhoc)$ term and eliminate $\J$ by means of (\ref{Raych}), to obtain a polynomial (with $\mu$-dependent coefficients) coupling the basic variables $\Ethree$ and $\O$ to the variable $\modU+\U_3^2$. A second polynomial in the same variables is arrived at by repeating this procedure on the real part of $\etd (\ref{Langevolsu})$. Elimination of $\Ethree$ from these two yields a relation, which is identical\footnote{The explanation being that all integrability conditions --including the ones used by Lang and Collins to arrive at their equations (8) and (9)-- have already been incorporated in the machinery to arrive at the present equations (\ref{e0u}-\ref{beq4}) and (\ref{beq2}-\ref{eq3}). However the last $\mu-p$ term in equation (9) of \cite{Lang} should be replaced by $\mu+3p$: as one can verify by the rest of the calculations in \cite{Lang}, this is just a typo.} (modulo a typo) to equation (9) of \cite{Lang},
\begin{eqnarray}
\theta' \ez \lambda^2 (\theta' \lambda^4 \ez+18 {p'}^2 \theta-7 p' \theta)(\modU+\U_3^2)
 +\frac{9 \theta}{\ez^2\lambda^2}[(9 {p'}-1)p' \theta-2 \theta' \epsilon] \O^2\nonumber \\
-\frac{3 \theta\theta'\ez}{2} (6  \theta \theta' \epsilon -2 \theta^2-9 \epsilon +6 \mu) = 0. \label{Langtrick1plus2}
\end{eqnarray}
Propagating (\ref{Langtrick1plus2})\footnote{Lang and Collins as well as Sopuerta use an extra commutator relation to arrive at a different second polynomial, but after elimination of $\modU+\U_3^2$, their result is identical with (\ref{Langtemp1}).} a similar polynomial is obtained, which, after elimination of the $\modU+\U_3^2$ term, leads to the following key equation, which, using (\ref{Langkey8},\ref{Langkey9}) to remove the higher order derivatives of $\theta$ and $p$, turns out to be identical with equation (43) of \cite{Lang}:
\begin{eqnarray}
\fl [45 \epsilon^2 (1+3 p') {\theta'}^4-45 \epsilon \theta p' (81 p'^2-36p'+7) {\theta'}^3+135 p'^3 \theta^2 (9 p'-1) {\theta'}^2] \O^2 \nonumber \\
\fl -\lambda^6 \ez^4 [(243 p' -9) \epsilon +(405 p'^2 -60 p'-5) \theta^2 +(15-135 p') \mu ] {\theta'}^4  +15 \lambda^{10} \ez^5 \theta (9 p'-1) {\theta'}^5 \nonumber \\
\fl -\lambda^2 \ez^3\theta p' [(729 p'^2 -648 p' +63) \epsilon+(35-135 p') \theta^2-(105-405 p') \mu] {\theta'}^3 =0 .\label{Langtemp1}
\end{eqnarray}
At this point one may follow the analysis of \cite{Lang}, but the present formalism allows for a more compact approach, particularly when dealing with the special case $p'=\tfrac{1}{9}$, $\theta'\neq 0$:\\

First notice that substitution of $\theta'= 0$ in (\ref{Langtrick1plus2}) implies $9 p'-1 = 0$ and hence, by (\ref{z_u_relation}), $\beta=0$, which was already excluded. Also notice that, when $p'$ is constant, necessarily
$9 p'-1 =0$: in fact, substituting $p'' = 0$ in (\ref{Langkey9}) yields $(9p'-1)(\epsilon \theta' -\theta p') =0$. The second factor cannot be 0, as its propagation would imply $p'= \tfrac{1}{3}$, after which integration of (\ref{Langkey8},\ref{Langkey9}) results in $\theta'= c_0 \lambda^{-3}$ ($c_0$ a non-0 constant). Substituting this in (\ref{Langtrick1plus2}) would give $3 \epsilon - 2 \mu =0$, in contradiction with $p'= -\tfrac{1}{3}$.\\
Continuing with $p'\neq \tfrac{1}{3}$ one eliminates $\mu$ from (\ref{Langtemp1}) and its propagation to obtain an equation which factorizes as $\mathcal{F}_1\cdot \mathcal{F}_2=0$, with
\begin{eqnarray}
\fl \mathcal{F}_1 &=& \epsilon \frac{\theta'}{\theta}+p'(18 p'-7) = 0, \\ \nonumber
\fl \mathcal{F}_2 &=& [20 \epsilon^3 (9 p'-1) {\theta'}^3-20 \epsilon^2 \theta (81 p'^3-18p'^2+1) {\theta'}^2-20 \epsilon \theta^2  p' (81 p'^3-144 p'^2+54 p'-7) \theta' \nonumber \\
\fl & & -20 p'^3 \theta^3 (9 p'-1)] \O^2 -6 \lambda^{14} \ez^6 (9 p'-1) {\theta'}^3-\lambda^{10} \ez^5 \theta (3 p'-1) (243 p'^2+27 p'+2) {\theta'}^2 \nonumber \\
\fl & &  +\lambda^6 \ez^4 p' \theta^2 (729 {p'}^3-27 p'-14) \theta' .
\end{eqnarray}
It is easy to check that $\mathcal{F}_1=0$ implies $p'= \tfrac{1}{3}$. On the other hand, propagating $\mathcal{F}_2=0$ and eliminating $\O$ results\footnote{this is the same procedure as followed in \cite{Lang}} in a polynomial $\mathcal{P}(\theta,\theta',\ez, \lambda, p')$ being 0, which, eliminating $\theta$ from $\mathcal{P}=\mathcal{P}'=0$, implies $p'$ is constant and hence, as shown above, $p'= \tfrac{1}{9}$.

This leaves us with only the case $\theta' \neq 0$ and $p'= \tfrac{1}{9}$ to be investigated\footnote{Note that one cannot use \cite{NorbertCQG1999} (in which the \emph{general} inconsistence of $p'=\tfrac{1}{9}$ was demonstrated) as \cite{NorbertCQG1999} uses the result of \cite{Lang} as a steppingstone.}: equation (\ref{Langtemp1}) simplifies then to
\begin{eqnarray}
\lambda (\theta - 3 \epsilon \theta') \O^2 +\frac{\ez^2}{90 \lambda}[\epsilon \theta'(81 \epsilon-30 \theta^2)+10 \theta^3-30 \mu \theta]=0,
\end{eqnarray}
which, after propagation and elimination of $\mu$ yields
\begin{eqnarray}
(9 \epsilon \theta'-5 \theta)(10 \O^2 -3 \lambda^2 \ez^3)(3 \epsilon \theta' - \theta)=0. \label{Lang123}
\end{eqnarray}
By propagation it is easy to see that the first factor of (\ref{Lang123}) being 0 is incompatible with $p'=\tfrac{1}{9}$. Substitution of the second factor in (\ref{Langtemp1}) shows that
\begin{equation}
10\O^2-3 \ez^3 \lambda^2=0 \textrm{ and } \mu=\tfrac{1}{3}\theta^2+\epsilon(\tfrac{9}{10}-\theta \theta'),
\end{equation}
by which (\ref{Langtrick1plus2}) would be reduced to $(9 \epsilon \theta'-5 \theta)(\modU+\U_3^2)=0$, thereby leading us back to the first case. This leaves one with the third possibility,
$3 \epsilon \theta' - \theta = 0$,
substitution of which in (\ref{Langtemp1}) gives
$\mu = \tfrac{9}{10}\epsilon$ (hence $\mu=9 p$), thereby simplifying (\ref{Langu}) and (\ref{Langtemp1}) to
\begin{eqnarray}
\u =-\tfrac{3}{2}i \beta (\ez^2\lambda^3 \theta')^{-1}, \\
10\ez^2(\modU+\U_3^2)+810\O^2-243 \ez^3 \lambda^2=0.
\end{eqnarray}
Herewith (\ref{Raych}) reduces to $\J - 9 \ez^2 \lambda^6{\theta'}^2=0$, implying $\J=9 c^2$ and $\theta'= c (\ez \lambda^3)^{-1}$ with $c$ a non-0 constant. By (\ref{Langu}) one has then $\u = -\tfrac{3}{2} i \beta (c \ez)^{-1}$, after which it follows from (\ref{Langkey1}) that $\lambda^2\ez^3$ is a constant given by
\begin{eqnarray}
216 c^2  \lambda^2 \ez^3 = [720 c^2 + 5(\rho + \rhoc)^2]\O^2  + 20 |\beta|^2 .
\end{eqnarray}
Together with the expressions obtained for $\mu, \theta'$ and $\u$ (\ref{eq2}) simplifies then to
\begin{eqnarray}
18 c^2 [3 i  \O (\rho + \rhoc)\Eone + \Ethree \beta + 3\betac \Etwo]\ez^2 + 144 i (18\O\kappa - 19\beta)\O c^3\ez \nonumber \\
-\beta [ 72\O^2 c^2 + 44 |\beta |^2 + 11 \O^2(\rho + \rhoc)^2] = 0.
\end{eqnarray}
This is a not-identically zero polynomial in $\ez$ with basic coefficients, implying that $\ez$ is basic and hence, as $\partial_0 (\ez) = \tfrac{2}{9} \ez \theta$, $\theta=0$.
\end{proof}

\section{The case of a Killing vector along the vorticity}\label{sectionKV}

When a Killing vector exists along the vorticity, then $\U_3=0$ and also ($\O, \theta, \J, \Ethree$ being invariants)
\begin{equation}
\ZZ(\O), \ZZ(\theta), \ZZ(\J), \ZZ(\Ethree)= 0.
\end{equation}
Substituting this in (\ref{e0U3},\ref{ZOmega}) it follows that
\begin{equation}
\rho=0,
\end{equation}
while (\ref{etU3}) implies
\begin{equation}\label{Killinglemma1}
\frac{3}{1+3p'} \Eone -\pi \uc =0.
\end{equation}
If $\pi\neq0 $ the propagation of (\ref{Killinglemma1}) along $\partial_0$ would result in the spatial gradients of $\theta$ and $\mu$ being related by
\begin{equation}
z = \frac{p'(3 G-2)}{1+3 p'} \u \theta,
\end{equation}
implying that  that $\theta=0$ or the existence of a non-constant function (of $\theta$ and $\mu$) with vanishing spatial gradient, such that $\omega=0$. Therefore, if a rotating and expanding shearfree perfect fluid exists and admits a Killing vector along the vorticity, then necessarily also\footnote{this is basically the argumentation given also in \cite{VdB-Radu2016}}
\begin{equation}
\pi=\Eone =0,
\end{equation}
after which
(\ref{eq7},\ref{eq3},\ref{eq10}) imply
\begin{equation}
\ZZ(\kappa)= \ZZ(\beta)=\ZZ(\Etwo)=0,
\end{equation}
The $[\ZZ,\, \et ]$ commutator applied to $\theta$ shows then $\ZZ(z)=0$ and hence, by (\ref{z_u_relation}), $\ZZ(\u)=0$ as well. It follows that the $\ZZ$ operator is identically 0 'when applied to everything in sight'\footnote{as Alan Held used to say}.\\
Note that, by (\ref{ZU3}) one obtains also the following useful algebraic relation,
\begin{eqnarray}
 -( \frac{ p'G}{1+3 p'}+\tfrac{1}{6}) \modU - \tfrac{1}{2}(\kappac \u+ \kappa \uc) + \frac{\theta^2}{6 \lambda^2} \nonumber \\
 + \frac{9 p' - 5}{2 \ez^2 (1 + 3 p')} \O^2 - \frac{\mu}{2 \lambda^2} - \frac{\J}{6} + \frac{1}{1 + 3 p'} \Ethree =0.\label{special_eq1}
\end{eqnarray}

Next I prove a minor generalization of theorem \ref{th3} %
of section \ref{section_recaptheorems},
which will be used in section \ref{section_app2}.

\begin{lm}\label{betakappalemma}
If a barotropic shear-free perfect fluid admits a Killing vector along the vorticity and if $\beta$ and $\kappa$ are 0, then $\omega \theta=0$.
\end{lm}

\begin{proof}
Substituting $\beta=\kappa=0$ in (\ref{beq3}) it follows that $\Etwo=0$. Dividing (\ref{eq5}) by $\u^2$ one obtains
\begin{eqnarray}
  -\frac{2 p'}{(1+3 p' )^2 }(-18 G {p'}^2+3 G^2-6 G p'  -2 G+3 (3 p'+1) \epsilon G')\theta \nonumber \\
 + \frac{i}{\ez(1+3 p' )} (54 G {p'}^2-54 {p'}^3+6 G p' +9 {p'}^2+6 p' -1) \O \lambda =0.
\end{eqnarray}
Eliminating $G', G$ from the real and imaginary part, a linear EOS is obtained, after which theorem \ref{th4} implies $\omega \theta =0$.
\end{proof}

\section{Application 4}\label{section_app2}
When a Killing vector exists which is aligned with the vorticity, then likely any proof of the conjecture may first have to deal with the subcase $\kappa=0$, because of the special form taken by the algebraic relations (\ref{eq7}-\ref{eq6}).  As a second application of the formalism I prove the following,

\begin{te}
If a barotropic shear-free perfect fluid admits a Killing vector along the vorticity and if $\kappa$ is 0, then $\omega \theta=0$.
\end{te}

As the validity of the conjecture is known to hold\cite{VdB-Radu2016} for a \emph{linear} equation of state (EOS), the goal is to arrive at an inconsistency by deducing the existence of a linear EOS from $\omega \theta \neq0$.

\begin{proof}
Imposing the conditions discussed in the previous section, together with $\kappa=0$, the basic equations imply
\begin{eqnarray}
\Etwo=0,\ \Ethree=\tfrac{2}{3}\J, \label{EJrelations} \\
\etd{\beta}=\et{\betac} = \mathfrak{B} \O ,
\end{eqnarray}
with $\mathfrak{B}$ an auxiliary real and 0-weighted basic variable, satisfying the integrability condition
\begin{equation}
\etd \et \beta -\et ( \mathfrak{B}\O) +\tfrac{2}{3}\J=0.
\end{equation}

First notice that necessarily $\beta\neq 0$: otherwise substitution of $\beta=\kappa=0$ in (\ref{eq5}) and dividing by $\u^2$, yields
\begin{eqnarray}
  -\frac{2 p'}{(1+3 p' )^2 }(-18 G {p'}^2+3 G^2-6 G p' +9 p'  \epsilon G'-2 G+3 \epsilon G')\theta \nonumber \\
 + \frac{i}{\ez(1+3 p' )} (54 G {p'}^2-54 {p'}^3+6 G p' +9 {p'}^2+6 p' -1) \O \lambda =0,
\end{eqnarray}
after which elimination of $G'$ and $G$ from the real and imaginary part, would lead to a linear EOS. Theorem 4 of \cite{VdB-Radu2016} would then imply $\omega \theta =0$.\\
Continuing with $\beta\neq0$ it follows from $\et \O = \beta$ that $\O$ can be assumed to be a \emph{non-constant} basic quantity, a fact which will be repeatedly used in the sequel.\\

The strategy will be to solve (\ref{eq6}) and its conjugate for $\u$ and $\uc$, which requires the determinant of the corresponding system to be non-0. As this determinant is proportional to $6 p'G  + 3 p' + 1$,
I first investigate the special case where this is 0.
\subsection{$6 p'G  + 3 p' + 1 =0$}
Using $\partial_0(6 p'G  + 3 p' + 1)=0$ to express $G, G'$ as functions of $p'$, the real part of (\ref{eq6}) tells us that
\begin{equation}
\modU = 6 \J \frac{7 {p'}+1}{(1+3 {p'})^2(6 {p'}+1)}-6 \O^2 \frac{81 {p'}^3-27 {p'}^2+45 {p'}+13}{(1+3 {p'})^2(6 {p'}+1)\ez^2} .\label{D0Uexpr}
\end{equation}
Propagating (\ref{D0Uexpr}) along $\partial_0$ and using the imaginary part of (\ref{eq6}) one obtains
\begin{eqnarray}
\fl  \betac \u &=& - i \frac{(11 {p'}+1) (27 {p'}+5)}{3 \lambda (1+3 {p'})^3(6 {p'}+1)}  \ez
\theta \J + i \frac{2673 {p'}^4-3240 {p'}^3+1080 {p'}^2+1116 {p'}+163}{3 \lambda \ez (1+3 {p'})^3  (6 {p'}+1)} \O^2 \theta \nonumber \\
\fl & &  +4 \frac{45 {p'}^3+3 {p'}^2-14 {p'}-2}{(1+3 {p'})^2(4 {p'}-1)(6 {p'}+1)} \O \J-\frac{1}{4 {p'}-1} \mathfrak{B}\O +2 \frac{\lambda^2 \ez}{4 {p'}-1} \O\nonumber \\
\fl & &  +2 \frac{243 {p'}^4-486 {p'}^3-198 {p'}^2+174 {p'}+43}{(1+3 {p'})^2(4 {p'}-1)(6 {p'}+1)\ez^2 }\O^3. \label{D0uexpr}
\end{eqnarray}
As $\theta$ is absent from the real part of (\ref{D0uexpr}), its propagation along $\partial_0$ leads to a kind of `equation of state', relating the $\mu$-dependent variables $p', \ez,\lambda$ to the basic variables $\O,\J$ and $\mathfrak{B}$:

\begin{eqnarray}
\fl -2 \ez^2 (7 {p'}+2) (6 {p'}+1) (1+3 {p'})^3 \mathfrak{B} +\ez^2 (60264 {p'}^5-11052 {p'}^4-10927 {p'}^3-1419 {p'}^2 \nonumber \\
\fl -5 {p'}+3)  \J +(-419904 {p'}^7+882576 {p'}^6-556065 {p'}^5-161811 {p'}^4+108936 {p'}^3\nonumber \\
\fl +29620 {p'}^2-1103 {p'}-521) \O^2+4 \lambda^2 \ez^3 (1+2 {p'}) (6 {p'}+1)^2 (1+3 {p'})^3 = 0. \label{D0EOS1}
\end{eqnarray}

Two similar `equations of state' can be obtained by successively propagating (\ref{D0EOS1}) along $\partial_0$, after which elimination of the basic variables leads to a polynomial equation in $p'$ (with constant coefficients) and hence to a linear EOS.

\subsection{$6 p'G  + 3 p' + 1\neq 0$}
When $6 p'G  + 3 p' + 1\neq 0$ (\ref{eq6}) is solved for $\u$, which is then subsituted in the propagation of the real part of (\ref{eq6}) along $\partial_0$ to obtain an expression for $\modU$ as a linear combination of the basic variables $\J, \mathfrak{B},\O^2$ with $\mu$-dependent coefficients:
\begin{eqnarray}
\fl C \modU + \left( 432 {p'} (5 {p'}-1) G^2+18 (2 {p'}-1) (72 {p'}^2-27 {p'}-1) G-1728 {p'}^3+678 {p'}^2-48 {p'}-6\right) \J \nonumber \\
\fl  -3 (6 G {p'}+3 {p'}+1) (12 G-1) (1+3 {p'})\mathfrak{B} -\frac{1}{{p'} \ez^2} \left( 1728 {p'}^2 (9 {p'}-1) G^2  +36 {p'} (864 {p'}^3 \right. \nonumber \\
\fl \left. -552 {p'}^2+41 {p'}+27) G-34992 {p'}^6+38880 {p'}^5  -12717 {p'}^4-432 {p'}^3+174 {p'}^2-48 {p'}+15 \right) \O^2 \nonumber \\
\fl  +6\frac{\ez \lambda^2}{{p'}}(6 G {p'} + 3 {p'} + 1) (12 G {p'} - 12 {p'}^2 - 2 {p'} + 1) (1 + 3 {p'})=0\label{U_eq1}
\end{eqnarray}
with
\begin{eqnarray}
\fl C = 9 (4 {p'}-1) (1+3 {p'}) (36 {p'}^2-19 {p'}-5) \epsilon G'+1728 G^3 {p'}^2+(1296 {p'}^3-810 {p'}^2+441 {p'}+45) G^2 \nonumber \\
\fl -(7776 {p'}^5-2808 {p'}^4-1422 {p'}^3-159 {p'}^2+138 {p'}+3) G-4 (1+3 {p'})^2 (1-3 {p'})^2  .\label{Ucoeff}
\end{eqnarray}
First notice that the coefficient $C\neq 0$. Otherwise one could use (\ref{Ucoeff}) to eliminate $G'$ from the first and second propagations along $\partial_0$ of (\ref{Ucoeff}) to obtain two more linear combinations of
$\J, \mathfrak{B},\O^2$ with $\mu$-dependent coefficients. Elimination of $\mathfrak{B}$ and $\J$ results then in an equation of the form
\begin{equation}
\mathit{P}_1 (G, p')  \O^2 + \mathcal{P}_2(G, p') =0,
\end{equation}
with $\mathcal{P}_1, \mathcal{P}_2$ constant coefficient polynomials of degree 9 in $G$ and degree 20 and 17 in ${p'}$ respectively.\\
$\O$ being non-constant basic and basic, it follows that $\mathcal{P}_1$ and $\mathcal{P}_2$ are identically 0. Elimination of $G$, using for example a resultant mod $p$ procedure (with $p$ a suitably chosen prime number\footnote{here $p=3$ suffices}), shows that $p'$ is a root of some non-identically 0 constant coefficient polynomial and hence a linear EOS results.\\

When $C\neq 0$ it follows from (\ref{U_eq1}) that $\modU$ is of the form
\begin{equation}\label{U_structure}
\modU = \mathcal{G}_1 \J +  \mathcal{G}_2\O^2 +  \mathcal{G}_3 \mathfrak{B} +  \mathcal{G}_4,
\end{equation}
with $ \mathcal{G}_1, \dots  \mathcal{G}_4$ functions of $\mu$ only. In order to obtain a second equation of this form, one eliminates $\et \J$ from $(\ref{eq2}, \ref{eq9})$ \footnote{having used the relations (\ref{EJrelations}) to express all appearing basic derivatives in terms of $\et \J$} to arrive at an equation which is linear in $\u,\modU$ and $\theta$. Multiplicating this with $\uc$ in order to obtain a 0-weighted equation, the imaginary part of it reads
\begin{eqnarray}
\fl \left( 3 (9 {p'}-13) (1+3 {p'}) (4 {p'}-1) \epsilon G'+(-648 {p'}^3+774 {p'}^2-273 {p'}+39) G^2+(342 {p'}^3+213 {p'}^2 \right. \nonumber \\
\fl \Bigl. +134 {p'}-41) G  +4 (-1+3 {p'})^2 (1+3 {p'})^2 \Bigr) \modU + \left(396 {p'}^2-402 {p'}+78) G-54 {p'}^2 +256 {p'}-58\right) \J \nonumber \\
\fl +\frac{\O^2}{{p'}^2 \ez^2}\left( -36 {p'} (72 {p'}^2 - 65 {p'} + 13) G + 2916 {p'}^5 - 5913 {p'}^4 + 3024 {p'}^3 - 1890 {p'}^2 - 52 {p'} + 91 \right) \nonumber \\
\fl - 3 (1 + 3 {p'}) (6 G {p'} + 3 {p'} + 1) \mathfrak{B}  +6 \lambda^2 \ez (1 + 3 {p'}) (6 G {p'} + 3 {p'} + 1) = 0 . \label{U_eq2}
\end{eqnarray}
Eliminating $\modU$ from (\ref{U_eq1},\ref{U_eq2}) results then in an `equation of state' of the form
\begin{equation}
\BB \mathcal{F}_1 +\J \mathcal{F}_2 + \O^2 \mathcal{F}_3 + \mathcal{F}_4 =0,\label{KV_EOS1}
\end{equation}
with $\mathcal{F}_1, \ldots \mathcal{F}_4$ defined by
\begin{eqnarray}
\fl \mathcal{F}_1 = 6 {p'} \ez^2 (1+3 {p'}) \Bigl( (9 (1 + 3 {p'}) (9 {p'} - 13) G - 3 (9 {p'} - 7) (1 + 3 {p'})^2) \epsilon G' \Bigr. \nonumber \\
\fl \left. -(486 {p'}^2 - 351 {p'} + 117) G^3 + (216 {p'}^2 + 216 {p'} + 144) G^2 \right. \nonumber \\
\fl \Bigl. + (1 + 3 {p'}) (162 {p'}^3 + 9 {p'}^2 - 24 {p'} - 23) G \Bigr) , \\
\fl \mathcal{F}_2 = 4 {p'} \ez^2  \left( (-27 (1 + 3 {p'}) (5 {p'} - 1) (9 {p'} - 13) G + 9 (1 + 3 {p'}) (135 {p'}^3 - 15 {p'}^2 - 145 {p'} + 33)) \epsilon G' \right. \nonumber \\
\fl \left.+ (7290 {p'}^3 - 4779 {p'}^2 - 351) G^3 + (8748 {p'}^4 - 20736 {p'}^3 + 4752 {p'}^2 + 594 {p'} + 234) G^2 \right. \nonumber \\
\fl  + (-16038 {p'}^5 - 405 {p'}^4 + 14580 {p'}^3 - 1692 {p'}^2 - 486 {p'} + 57) G \nonumber \\
\fl \Bigl. + 8 (9 {p'} - 4) (-1 + 3 {p'})^2 (1 + 3 {p'})^2 \Bigr), \\
\fl \mathcal{F}_3 = 432 {p'} (9 {p'}-13) (1+3 {p'}) (9 {p'}-1) G-36 (1+3 {p'}) (1755 {p'}^3-1473 {p'}^2-43 {p'}+65)) \epsilon G' \nonumber \\
\fl -432 {p'} (486 {p'}^3-297 {p'}^2-13) G^3+(-314928 {p'}^5+624024 {p'}^4-144180 {p'}^3-9396 {p'}^2 -1692 {p'}\nonumber \\
\fl -2340) G^2+(472392 {p'}^7-472392 {p'}^6+624024 {p'}^5-262764 {p'}^4-72576 {p'}^3+2232 {p'}^2+864 {p'}+444) G\nonumber \\
\fl -8 (729 {p'}^4-810 {p'}^3+270 {p'}^2-18 {p'}-19) (-1+3 {p'})^2 (1+3 {p'})^2, \\
\fl \mathcal{F}_4 = 3 \lambda^2 \ez^3)(1+3 {p'}) \left((-36 (9 {p'}-13) {p'} G (1+3 {p'})+3 (72 {p'}^2-89 {p'}+13) (1+3 {p'})^2) \epsilon G' \right. \nonumber \\
\fl \left. +36 {p'} (54 {p'}^2-39 {p'}+13) G^3+(-1944 {p'}^4+810 {p'}^3-909 {p'}^2-732 {p'}+39) G^2\right. \nonumber \\
\fl \left. -(1+3 {p'}) (648 {p'}^4-306 {p'}^3-309 {p'}^2-226 {p'}+41) G+4 (-1+3 {p'})^2 (1+3 {p'})^3 \right) .
\end{eqnarray}

First notice that the following expression $\Sigma$ is non-0 i) when $\mathcal{F}_1\mathcal{F}_2 = 0$ (as in cases I,II and II below) and ii) when $\mathcal{F}_1\mathcal{F}_2 \neq 0$ and $\F_2=c_2 \F_1$ or $\F_3=c_1 \F_1-c_2 \F_2$ (as in cases IV or V below):

\begin{equation}
\fl \Sigma = 18 (9 {p'} - 13) {p'} G^2 + 3 (1 + 3 {p'}) (81 {p'}^3 - 153 {p'}^2 + 88 {p'} - 12) G - (9 {p'} - 7) (-1 + 3 {p'})^2 (1 + 3 {p'})^2 \label{Sigmadef}
\end{equation}
In fact, propagating $\Sigma =0$ and eliminating $\epsilon G'$ (using either $\mathcal{F}_1=0$, $\mathcal{F}_2=0$, $\F_2=c_2 \F_1$ or $\F_3=c_1 \F_1-c_2 \F_2$ together with its propagation) leads to a linear EOS.
\\

As $\O$ is non-constant and basic, one has $\textrm{c-rank}(1,\O^2)= 2$.  By lemma \ref{lincolemma} and (\ref{KV_EOS1}) it follows that $\textrm{c-rank}(1,\O^2,J,\BB) = 2 \textrm{ or } 3$, which leads us investigate the following cases:
\begin{enumerate}[I)]
\item $\F_1=\F_2=0$,
\item $\F_1=0\neq\F_2$,
\item $\F_2=0\neq\F_1$,
\item $\F_1\F_2 \neq 0$ and $\F_2=c_2 \F_1, \, \F_3=c_3 \F_1, \F_4=c_4 \F_1$,
\item $\F_1\F_2 \neq 0$ and $\F_3=c_1 \F_1-c_2 \F_2, \, \F_4=c_3\F_1-c_4\F_2$,
\end{enumerate}
with $c_1,c_2,c_3,c_4$ constants (necessarily real).
\begin{enumerate}
\item[\bf{Case I}] As $\O$ is not constant this implies also $\F_3=\F_4=0$. Eliminating $\epsilon G'$ and $G$ from $\F_1,\F_2,\F_3$ immediately yields a linear EOS.
\item[\bf{Case II}]  By lemma \ref{lincolemma} $\textrm{c-rank}\,(1,\O^2,\J)=2$ and hence $\textrm{c-rank}\,(\F_2,\F_3,\F_4)=1$. Therefore constants $c_3,c_4$ exist (both non-0, as otherwise again a linear EOS results) such that $\F_3-c_3 \F_2= \F_4 - c_4 \F_2 =0$. Eliminating $\epsilon G'$ and $G$ from these and from $\F_2=0$ one obtains $\mathcal{P}_1=0$, with $\mathcal{P}_1$ a polynomial of 4th degree in $\lambda$ and coefficients in $\ez$ and $p'$. Propagating the latter and eliminating $\lambda$ the result factorizes to give us two polynomial equations, $\mathcal{P}_2=0$ or $\mathcal{P}_3=0$: propagating either of these to eliminate $\ez$ again $p'$ is recognized as the root of a polynomial with constant (non-0) coefficients, leading to a linear EOS.
\item[\bf{Case III}] For the same reasons as above now constants $c_3,c_4$ exist (both non-0, as otherwise again a linear EOS results) such that $\F_3-c_3 \F_1= \F_4 - c_4 \F_1 =0$, together with a basic relation $\BB+c_3 \O^3 +c_4 \O = 0$. The analysis is slightly more complicated than the one in case II, but essentially proceeds along the same lines and implies again a linear EOS.
\item[\bf{Case IV}] We still have three $\F$-relations to proceed with the successive elimination of $\epsilon G'$ and $G$, after which twice propagating the remaining equation allows one (dividing out the non-0 factor (\ref{Sigmadef})) to obtain once again a polynomial equation in $p'$ with non-0 coefficients and hence a linear EOS.
\item[\bf{Case V}] This is the case with $\F_1\F_2\neq 0$ and $\textrm{c-rank}(1,\O^2,J,\BB) = 2$ and is the most complicated, as now only two $\F$-relations are available (and, furthermore, for which elimination of $G'$ does not seem to lead anywhere). Now however the $\textrm{c-rank}$ of the basic coefficients is 1,
    \begin{eqnarray}
    \J=c_2\O^2+c_4,\\
    \BB=-c_1\O^2-c_3,
    \end{eqnarray}
 which can be used to obtain
 \begin{equation}
 \et(\J) = 2 c_2\O\beta,\label{case5_Jrel}
 \end{equation}
 thereby simplifying the algebraic equation (\ref{eq2}). Dividing the latter by $u$ one obtains a 0-weighted relation, the imaginary part of which reads,
 \begin{equation}
 \fl -5 i  \lambda (1+3 {p'}) G \modU \theta-1/4 \lambda^2 (1+3 {p'}) (4 {p'} c_2 \ez^2-63 {p'}-\
13) (\betac u-\beta \uc)/\ez/{p'} = 0.
 \end{equation}
 Substituting in this the earlier expressions (\ref{D0uexpr},\ref{U_eq2}) for $u$ and $\modU$, one obtains a quadratic equation in $\O$ with coefficients depending on $G',G,\lambda,\ez,p'$ and the constants $c_1, \ldots c_4$. As $\O$ is basic and not constant, it follows that all coefficients must vanish: eliminating then $G'$ from the $\O^2,\O^0$ coefficients, one obtains a polynomial which factorizes as $\mathcal{P}_1 \mathcal{P}_2\mathcal{P}_3=0$, with
 \begin{eqnarray}
 \fl \mathcal{P}_1 &=& p'(c_2\ez^2 - 63) - 13   \\
 \fl \mathcal{P}_2 &=& 3 {p'} (36 {p'}^2 c_2 \ez^2-37 {p'} c_2 \ez^2+5 c_2 \ez^2-432 {p'}^2+282 {p'}
+14) G\nonumber \\
\fl & & -(1+3 {p'}) (-1+3 {p'})^2 (4 {p'} c_2 \ez^2-63 {p'}-13) \\
\fl \mathcal{P}_3 &=& (-12 {p'} \ez (c_2 \ez^2-6) \lambda^2+6 {p'} (c_1 c_4 \ez^2-c_2 c_3
\ez^2+6 c_3-4 c_4)) G\nonumber \\
\fl & & -2 \ez (-4 {p'} c_2 \ez^2+81 {p'}^3-27 {p'}^2+27 {p'}+
7) \lambda^2-4 {p'} c_1 c_4 \ez^2+4 {p'} c_2 c_3 \ez^2-81 {p'}^3 c_3\nonumber \\
\fl & & -270 {p'}^3 c_4+
27 {p'}^2 c_3+234 {p'}^2 c_4-27 {p'} c_3-90 c_4 {p'}-7 c_3+14 c_4 .
 \end{eqnarray}
 A tedious but straightforward process of propagation of each of these polynomials and eliminating $G', G, \lambda, \ez$ from the resulting expressions and from $\F_3=c_1 \F_1-c_2 \F_3, \, \F_4=c_3\F_1-c_4\F_2$ eventually leads again to a linear EOS.
\end{enumerate}
\end{proof}

The interested reader can obtain Maple worksheets with full details of all the proofs from the author.

\section{Appendix 1}
Complexified equations for acceleration and expansion (see section \ref{complexification}):

\begin{eqnarray}
 \partial_0(\u) &=& p' \u \theta +z,\label{ap0u} \\
 \d (\u)- i \tau \u  &=& \frac{9 G p' + 3 p' + 1}{3 (1 + 3 p')}  \u^2  + 2 \pi  \U_3 + \frac{3}{1 + 3 p'}\Etwo ,\label{apdu}\\
 \dc (\u) - i \tauc  \u &=&  \frac{G p'}{1 + 3 p'}\modU - \frac{6 G p' + 3 p' + 1}{3 (1 + 3 p')} \U_3^2 + i \rho \U_3 + \frac{\theta^2}{3 \lambda^2}  \nonumber \\
  & & - 2 i \frac{\O\theta}{\lambda \ez} + \frac{9 p' + 7}{\ez^2 (1 + 3 p')} \O^2 - \frac{\mu}{\lambda^2} - \frac{\J}{3} - \frac{\Ethree}{1 + 3 p'} ,\label{apdcu}\\
 \Z(\u) - \frac{i \nu}{2} \u &=& \frac{9 G p'+3 p'+1}{3(1+3 p')} \U_3 \u +\kappa \U_3+\frac{3}{1+3 p'} \Eone,\label{apZu}
\end{eqnarray}

\begin{eqnarray}
\partial_0(\U_3) &=& p' \U_3 \theta + \frac{3 \lambda}{4 \epsilon_0} \O (\rho+\rhoc),\label{ap0U3} \\
\d (\U_3) &=& \left( \frac{9 G p'+3 p'+1}{3 (1+3 p')} \U_3 + i \frac{\rhoc}{2}  \right) \u- \pi \uc +\frac{3}{1+3p'} \Eone, \label{apdU3}\\
\Z (\U_3) &=& -\frac{6 G p' + 3 p' + 1}{6 (1 + 3 p')} \modU - \tfrac{1}{2}(\kappac \u+ \kappa \uc)
+ \frac{12 G p' + 3 p' + 1}{6 (1 + 3 p')} \U_3^2 \nonumber \\
& & + \frac{\theta^2}{6 \lambda^2} + \frac{9 p' - 5}{2 \ez^2 (1 + 3 p')} \O^2 - \frac{\mu}{2 \lambda^2} - \frac{\J}{6} + \frac{1}{1 + 3 p'} \Ethree,\label{apZU3}
\end{eqnarray}

\begin{eqnarray}
\partial_0(\theta) &=& \frac{\lambda^2 p'}{2}(1-2 G)(\modU +\U_3^2) +(\frac{p'}{2}-\frac{1}{3})\theta^2+\frac{\lambda^2}{2 \ez^2}(9p'+4)\O^2 \nonumber \\
& & -\frac{\lambda^2 p'}{2} \J -\frac{3}{2} \lambda^4\ez +(1-\frac{3}{2}p')\mu  ,\\
\d (\theta) &=& z ,\\
\Z (\theta) &=& \frac{3 \lambda}{4\ez } \O (\rho+\rhoc),\label{apZtheta}
\end{eqnarray}

\section{Appendix 2}
Basic equations (see section \ref{complexification}):
\begin{eqnarray}
\d (\O) = \beta -\kappa \O, \label{apdO}\\
\Z (\O) = \frac{i}{2} \O (\rho - \rhoc),\\
\d (\rho)+2 i \dc (\pi)+4 \pi \tauc = -(\rho+\rhoc ) \kappa-2 i \Eone ,\\
\d (\kappa)-2 \Z(\pi)-i \kappa \tau+2 i \nu \pi = i (\rhoc -\rho)\pi-\kappa^2-\Etwo ,\\
\dc (\kappa)-i \kappa \tauc-i \Z(\rho) = -| \kappa |^2-2 | \pi |^2 +\tfrac{2}{9} \J-\tfrac{1}{3} \Ethree+\tfrac{1}{2} \rho^2 ,\\
\d (\nu)-2 \Z(\tau) = -2 \pi \tauc+i (\rhoc-\nu) \tau -( \rhoc+\nu ) \kappa +2 i ( \pi \kappac- \Eone) ,\\
\dc (\tau)-\d (\tauc) = i \left( \tfrac{1}{2}(\rho+ \rhoc) \nu+2 | \tau |^2-2 | \pi|^2 +\tfrac{1}{2} | \rho|^2-\tfrac{2}{9} \J-\tfrac{2}{3} \Ethree \right). \label{Beq1}
\end{eqnarray}

\section{Appendix 3 (integrability conditions)}

From the basic equations (\ref{dOmega},\ref{ZOmega}) and the $[\ZZ,\, \et] \O$, $[\etd,\, \et] \O$ commutator relations one finds
\begin{eqnarray}
\fl \etd(\beta)-\et(\betac)-i \O\ZZ(\rho+\rhoc)=\betac \kappa - \beta \kappac,\label{beq2}\\
\fl \O \, \et(\rho-\rhoc)+2 i \left(\ZZ(\beta)-\O \ZZ(\kappa)\right)=
2 i \pi(\betac-\O \kappac)-\rho (\beta+\O \kappa)+2 \beta \rhoc . \label{beq6}
\end{eqnarray}

\noindent Next come the integrability conditions for $\u$ and $\theta$,\\

\noindent $[\partial_0,\, \ZZ] \u$:

\begin{eqnarray}
\fl 2 \frac{\lambda \ez p'}{1+3 p' }  (18 G {p'}^2-3 G^2+6 G p' -9 p'  \epsilon G'+2 G-3 \epsilon G') \u
\theta \U_3 \nonumber \\
\fl +\frac{i}{3} \lambda^2 (1+3 p' ) (27 G p' -27 {p'}^2+15 p' -2) \O \U_3 \u \nonumber \\
\fl +\tfrac{1}{2} \lambda^2 (9 G \rhoc p' +9 G p'  \rho-36 \rhoc {p'}^2+18 \rho {p'}^2-6 \rhoc
 p' +6 \rho p' +2 \rhoc) \O \u \nonumber \\
 \fl -i \lambda^2 (1+3 p' ) (9 p' -1) \kappa \O \U_3+3
 i \lambda^2 \beta (3 G p' -6 {p'}^2+p' +1) \U_3+6 \frac{\lambda \ez p'}{1+3 p'}  (3 G-2)
\Eone \theta\nonumber \\
\fl -3 i (9 p' -1) \lambda^2 \O \Eone+\tfrac{3}{2} \lambda^2 \kappa (1+3 p' )
 (\rho+\rhoc) \O-3 i \lambda^2 (1+3 p' ) \ZZ(\beta) = 0. \label{eq10}
\end{eqnarray}

\noindent $[\partial_0,\, \et] \u$:

\begin{eqnarray}
\fl  \frac{2 p'}{(1+3 p' )^2 }(18 G {p'}^2-3 G^2+6 G p'+2 G-3(1+3 p') \epsilon G')\theta \u^2 \nonumber \\
\fl + \frac{i}{\ez(1+3 p' )} (54 G {p'}^2-54 {p'}^3+6 G p' +9 {p'}^2+6 p' -1) \O \lambda \u^2 \nonumber\\
\fl + i \frac{\lambda (9 p' -1)}{\ez} \kappa \O \u+i \frac{\lambda}{\ez (1+3 p') } (18 G p' -54 {p'}^2-3 p' +5) \beta  \u +6 \frac{p'}{(1+3 p' )^2}  \Etwo (3 G-2) \theta \nonumber\\
\fl -2 i \frac{(9 p' -1)\lambda}{\ez} \O \pi \U_3-3 i \frac{(9 p' -1) \lambda}{\ez(1+3 p' )} \Etwo \O+3 (\rho+\rhoc) \pi \frac{\lambda}{\ez} \O-3 i  \frac{\lambda}{\ez}\et(\beta) = 0, \label{eq5}
\end{eqnarray}

\noindent $[\partial_0,\, \et] \uc$:

\begin{eqnarray}
\fl i (3 G p' +18 {p'}^2+3 p' -1) \beta \lambda \uc -i (3 G p' -18 {p'}^2+2) \lambda \betac u-i (1+3 p' ) \lambda (9 p' -1) \kappa \uc \O \nonumber \\
\fl -\tfrac{2}{3} i (45 G {p'}^2+21 G
p' +9 {p'}^2-1) \lambda \modU \O  -\tfrac{i}{3}(90 G {p'}^2+6 G p' +9 {p'}^2-1) \lambda {\U_3}^2 \O\nonumber \\
\fl -\tfrac{2}{3} \frac{\ez}{1+3 p'} p'  (-18 G {p'}^2+3 G^2-6 G p' +9
 p'  \epsilon G'-2 G+3 \epsilon G') (\modU -2 \U_3^2) \theta \nonumber \\
 \fl -(6 G ( \rhoc +\rho ) p' -27 \rhoc {p'}^2 +3 (\rho - \rhoc ) p'
 +2 \rhoc+\rho) \lambda \O {\U_3} +\tfrac{5}{3} i \frac{(1+3 p' ) (-1+3 p' )}{\lambda }\O \theta^2 \nonumber \\
\fl -2\frac{\ez}{1+3 p'} (3 G-2) p'  \theta \Ethree +\tfrac{2}{3}\frac{1}{\ez(1+3 p' )} (-81 {p'}^3+36 G
 p' +27 {p'}^2-27 p' -7) \theta \O^2\nonumber \\
 \fl  +\frac{i \lambda}{\ez^2 }(135 {p'}^2+96 p' +1) \O^3
 -\tfrac{i}{2} \frac{1+3 p'}{\lambda} (12 \epsilon +3 \rhoc (\rhoc + \rho ) \lambda^2+30 \mu p' -10 \mu) \O\nonumber \\
\fl  -i (9 p' -1) \lambda \Ethree \O +3 i (1+3 p' ) \lambda \et(\betac)-\tfrac{i}{3}(1+3p')\lambda(15 p'-1)\O\J = 0, \label{eq6}
\end{eqnarray}

\noindent $[\ZZ,\, \et] \u$:

\begin{eqnarray}
\fl 18 {p'}^2 (\U_3 \Etwo-\u \Eone)+2 i \O \pi \frac{(1+3 p' )^2}{\lambda \ez} \theta+\tfrac{3}{2} i (1+3 p' )\rhoc \Etwo +3 \pi (1+3 p' ) \Ethree \nonumber \\
\fl -12  \frac{1+3 p'}{\ez^2}\pi \O^2+3 (1+3 p' ) \left(\ZZ(\Etwo)-\et(\Eone) -2 \kappa \Eone \right) = 0,\label{eq7}
\end{eqnarray}

\noindent $[\ZZ,\, \et] \uc$:

\begin{eqnarray}
\fl 6\frac{{p'}^2}{1+3 p'}(\Eone \uc - \Eonec \u )-\tfrac{4}{9} i \frac{9 G p' -9 {p'}^2+1}{\ez \lambda} \O \theta \U_3+ \frac{1+3 p'}{3 \lambda \ez} (\rhoc-
\rho) \O \theta \nonumber \\
\fl -\tfrac{i}{2} (\rho+\rhoc) \Ethree+3 i (\rho+\rhoc) \frac{p' +1}{\ez^2} \O^2+\pic \Etwo-\pi \Etwoc -\et(\Eonec)+\etd(
\Eone) = 0, \label{eq8}
\end{eqnarray}

\noindent $[\etd,\, \et] \u$:

\begin{eqnarray}
\fl 18 \lambda \ez {p'}^2 \Ethree \u -2 i (9 G p' -9 {p'}^2+1) (1+3 p' ) \O \u \theta
-3  \frac{\lambda (p' +1) (81 {p'}^2-5)}{\ez} \O^2 \u \nonumber \\
\fl +18 \lambda \ez {p'}^2 (\Etwo \uc-2 \Eone \U_3) -4 i \kappa (1+3 p' )^2 \O
 \theta+2 i \beta (1+3 p' )^2 \theta\nonumber \\
 \fl +12\frac{\lambda (1+3 p' )}{\ez} \kappa \O^2-3 \frac{\lambda (9 p' +11) (1+3 p')}{\ez} \beta  \O  \nonumber \\
 \fl+3 \lambda \ez (1+3 p' ) \left(\kappa \Ethree -\kappac \Etwo+6 \pi \Eonec -2 i \rho \Eone -3 i \rhoc \Eone -2 \ZZ(\Eone)+\et(\Ethree)+\etd(\Etwo) \right) =
0, \label{eq9}
\end{eqnarray}

\noindent $[\partial_0,\, \ZZ] \theta$:

\begin{eqnarray}
\fl \tfrac{3}{2} \ez p'  \lambda \frac{-1+2 G}{1+3 p'}( \Eonec \u+ \Eone \uc )+\tfrac{1}{2}  \frac{\ez p'}{1+3 p'}  (2 G^2 p' -2 G^2-G
p' -6 p'  \epsilon G'+G-2 \epsilon G') \lambda \U_3 ( \modU+{\U_3}^2)\nonumber \\
\fl +\tfrac{5}{6} \frac{G p'  \ez}{\lambda} \U_3 \theta^2+(-1+
2 G) \frac{\ez p'  \lambda}{1+3 p' } \U_3 \Ethree -\tfrac{5}{6} \lambda \ez p'  G \U_3 \J \nonumber \\
\fl +\tfrac{1}{6}  \frac{\lambda}{\ez(1+3 p' } (135 G {p'}^2-162 {p'}^3-3 G p' -108
 {p'}^2+54 p' +16) \O^2 \U_3 \nonumber \\
 \fl -\tfrac{1}{2} \frac{\ez}{\lambda}(3 p'  \lambda^4 \ez+\lambda^4 \ez+5 G p'  \mu)  \U_3+
\tfrac{1}{2} i \frac{ \lambda}{\ez}(-\rho+\rhoc) (9 p' +4) \O^2+\tfrac{1}{2} \ZZ(\J) \lambda
\ez p'  = 0, \label{eq1}
\end{eqnarray}

\noindent $[\partial_0,\, \et] \theta$:

\begin{eqnarray}
\fl \tfrac{1}{2} \lambda^2 (2 G^2 p' -2 G^2-G p' -6 p'  \epsilon G'+G-2 \epsilon G') \u (\modU+\U_3^2)
+\tfrac{5}{6} G (1+3 p' ) \u \theta^2-\tfrac{5}{2} i \frac{G \lambda (1+3 p' )}{\ez} \O  \u \theta \nonumber \\
\fl +\tfrac{1}{2} (-1+2 G) \lambda^2(3 \Etwo \uc  -\Ethree \u)+\frac{1}{12}\frac{ \lambda^2}{\ez^2 p'} (270 G {p'}^2-1053 {p'}^3+138 G p' -297 {p'}^2+81 p' +29)\O^2 \u \nonumber \\
\fl -\tfrac{5}{6} G \lambda^2 (1+3 p' ) \u \J-\frac{1+3 p' }{2 p'} (3 p'  \lambda^4 \ez+
\lambda^4 \ez+5 G p'  \mu) \u+3
\Eone (-1+2 G) \lambda^2 \U_3\nonumber \\
\fl +\kappa \lambda^2 \frac{(9 p' +4) (1+3 p' )}{\ez^2 p'} \O^2
-\beta \lambda^2 \frac{(1+3 p' ) (63 p' +13)}{4 \ez^2 p'}  \O+\tfrac{1}{2} (1+3 p' ) \lambda^2 \et(\J) = 0,  \label{eq2}
\end{eqnarray}

\noindent $[\ZZ,\, \et] \theta$:

\begin{eqnarray}
\fl \frac{1}{3(1+3 p')}  (36 G p'+9 p'^2-1) \O \U_3 \u+ i (\rhoc+\tfrac{3}{2} \rho -\tfrac{3}{2} \rhoc p'-6 \rho p') \O u \nonumber \\
\fl- (9 p'-1)  \O (\pi \uc+ \kappa \U_3)  - (3 p'-2) \beta \U_3-3  \frac{15 p'+1}{1+3 p'} \O \Eone+\tfrac{3}{2}  (2 i \kappa \rhoc+ i \kappa \rho+2   \pi \kappac) \O \nonumber \\
\fl - 3  i \rho \beta -6 \pi \betac -3  \O \left(\ZZ(\kappa)+2 \etd(\pi)\right) = 0, \label{eq3}
\end{eqnarray}

with $[\etd,\, \et] \theta$ being identical to the real part of (\ref{eq6}).\\

\noindent There remain the integrability conditions for $\U_3$, which under the previous algebraic relations become identities, except for $[\ZZ,\, \et] \U_3$ and $[\etd,\, \et] \U_3$, which reduce to the following basic equations,
\begin{eqnarray}
\fl 18 \ZZ(\Eone)+9 \etd(\Etwo)-3 \et(\Ethree)+2 \et(\J)+9 i  (\rhoc-2 \rho) \Eone+9 \kappac \Etwo-18 \Eonec \pi
-9 \Ethree \kappa = 0, \label{beq8}\\
\fl 2 \ZZ(\J)+6 \ZZ(\Ethree)+9 (\et(\Eonec)+ \etd(\Eone))-\tfrac{9}{2} i  (\rho-\rhoc) \Ethree+18 (\kappac \Eone+ \kappa \Eonec)\nonumber \\
\fl +9 (\pic \Etwo+ \pi \Etwoc ) = 0 .\label{beq9}
\end{eqnarray}

A more compact form of the $\theta$ commutator relations can be obtained by refraining from writing out the second derivatives appearing in the commutators. One obtains then

$[\partial_0,\, \ZZ] \theta$:

\begin{eqnarray}
\ZZ(\Theta) &=& -p' \U_3 \Theta+\tfrac{1}{4} \frac{\lambda (3 p'-2)}{\ez} \O (\rho+\rhoc) \theta \label{apsol1}
\end{eqnarray}

$[\partial_0,\, \et] \theta$:

\begin{eqnarray}
\et(\Theta) &=& p' [i \frac{\lambda (2-9 G)}{2\ez}\O \theta - \Theta ] \u +[i \frac{(9p'-1)\lambda}{2\ez} \O + \frac{3p'-2}{3} \theta ] z \label{apsol2}
\end{eqnarray}

$[\ZZ,\, \et] \theta$:

\begin{eqnarray}
\ZZ(z) &=& -\tfrac{1}{4} \frac{(6 p'+1) \lambda}{\ez} (\rho+\rhoc) \O \u- i (\rhoc+\half \rho) z+ \pi \zc +\tfrac{3}{4} \frac{\lambda}{\ez} \et(\rho+\rhoc) \label{apsol3}
\end{eqnarray}

$[\etd,\, \et] \theta$:

\begin{eqnarray}
\etd(z)-\et(\zc) & & + i \frac{\O}{\lambda \ez}[ 4 \Theta-\tfrac{3}{4} (\rho+\rhoc)^2 \lambda^2 ] = 0 \label{apsol4}
\end{eqnarray}

Of course, subsitution of (\ref{Raych}) and (\ref{e0z}-\ref{ZZz}) in these equations reduces them to the explicit forms (\ref{eq1}-\ref{eq3}) and $\Re(\ref{eq6})$.

Also note that, as expected, all equations (\ref{beq2}-\ref{apsol4}) are well-weighted. \\

\section*{Acknowledgement}

All calculations were performed or checked with Maple 2022. \\
I thank John Carminati for noticing a number of typos in earlier versions of this document.

\end{document}